\documentclass{svproc}
\PassOptionsToPackage{hyphens}{url}
\usepackage{graphicx}
\usepackage{multirow}
\usepackage{amsmath, amssymb, amsfonts}
\usepackage{mathrsfs}
\usepackage[title]{appendix}
\usepackage{xcolor}
\usepackage{textcomp}
\usepackage{manyfoot}
\usepackage{booktabs}
\usepackage{algorithm}
\usepackage{algorithmicx}
\usepackage{algpseudocode}
\usepackage{paralist}

\newcommand{\properties}{Platform Influence Functions Properties}
\renewcommand{\epsilon}{\varepsilon}
\newcommand{\opiniondynamic}{opinion dynamic}

\usepackage[breaklinks]{hyperref}  

\begin{document}

\pagestyle{plain}  
\mainmatter
\title{Harmonizing vs Polarizing Platform Influence Functions} 
\titlerunning{HPPIF}  
\author{Hind AlMahmoud , Frederik Mallmann-trenn\thanks{Part of the research was supported by the EPSRC grant EP/W005573/1.}}
\authorrunning{Anonymous et al.} %

\institute{King's College London\\
\email{hind.almahmoud@kcl.ac.uk,frederik.mallmann-trenn@kcl.ac.uk},\\}

\maketitle              

\begin{abstract}

We investigate the dynamics of opinion formation on social networking platforms, focusing on how individual opinions, influenced by both social connections and platform algorithms, evolve. We model this process using a differential equation, accounting for both peer influence and the platform's content curation based on user preferences. Our primary aim is to analyze how these factors contribute to opinion polarization and identify potential strategies for its mitigation. We explore the conditions under which opinions converge to a consensus or remain polarized, emphasizing the role of the platform's influence function. Our findings in two-agent, complete graphs, and stochastic block model provide insights into the impact of social media algorithms on public discourse and offer a framework for understanding how polarization can be avoided.
\end{abstract}

\section{Introduction}\label{sec1}

Research shows that platforms and media can significantly influence people's opinions. For instance, 14\% of Americans reported changing their views due to something they saw on social media, with younger users being more affected \cite{pew2018}, which highlights the influential role of media. Furthermore, social networks play a critical role in shaping political decisions and disseminating opinions, as demonstrated during events like the 2008 U.S. presidential election \cite{burbach2020opinion}. As these platforms have become integral to our daily interactions, the content they present—often tailored to individual user preferences—can substantially shape our views. This customization tends to favor content that aligns with users' existing beliefs, which can potentially lead to increased polarization. In this paper, we examine the impact of social networking platforms, such as Instagram and Facebook, on the formation of opinions and beliefs.
Our research aims to deepen the understanding of this polarization process and explore potential avenues for its mitigation. 
We approach this by modeling a network of users, where each individual at any given time \( t \) has an opinion \( x_i(t) \) within a range of \([-1,1]\). Each opinion evolves based on an equation that considers both the influence of others and the platform's algorithmic impact, represented as:
\begin{equation}\label{eq:maineq}
    \dot{x_i}(t) =  \sum_{j \neq i}^{n} a_{i j}(x_j(t) - x_i(t)) + b \cdot (f(x_i(t)) - x_i(t)),
\end{equation}
where \( f(\cdot) \) is the platform's influence function since platforms like Facebook display content (posts) using algorithms that typically align with the user's existing opinions and preferences\cite{garcia2023effect}, $n$ is the number of users and \( b \) is a positive constant.
We denote the vector of functions $x$ as an \emph{\opiniondynamic} if it starts at a feasible point (in $[-1,1]^n$) and then evolves as described by the ODE \eqref{eq:maineq}.
This model was introduced by \cite{candogan2022social}, who studied $f(\cdot)$ for the sign function. The influence parameter \(a_{ij} \geq 0\) can take any value greater than zero. For simplicity, we assume that \(a_{ij} = a_{ji}\), i.e., the matrix \(A = (a_{ij})\) is symmetric.

Our paper builds on existing models to analyze how the platform function affects opinion dynamics, particularly focusing on the system's convergence to stable states and the potential of the platform function to either exacerbate or mitigate polarization. There are two types of final states persistent disagreement states in which the agents will never reach consensus, with their opinion unchanging, and consensus, where agents align on a common opinion, meaning their final values share the same sign. Strong consensus is a more specific case where all agents not only agree in general but hold exactly the same opinion.

\subsection{Contributions}
Under mild assumptions on $f(\cdot)$, such as continuity, symmetry, sign preservation, and boundedness, (see Section~\ref{sec:properties}) we give sufficient and necessary conditions for persistent disagreement states. See Theorem~\ref{thm:B},~\ref{thm:C}, and~\ref{thm:BN}. In addition, we characterize all the states where strong consensus is reached in Theorem~\ref{thm:A} and~\ref{thm:AN}. Also, we generalise our results to stochastic block model of two blocks as seen in section~\ref{sec:SBM}. 
Finally, we analyse a set of harmonizing and polarizing functions in Section~\ref{sec:specificResults},  where we show for a function $f(\cdot)$ that violate sign preservation, that the system always converges to consensus see Theorem ~\ref{thm:antiSignFunction}. 

\subsection{Related Literature}
Over the past fifty years, the development of mathematical models to analyze social learning and opinion dynamics has attracted significant attention. Most of the work is based on DeGroot model~\cite{degroot1974reaching} and Friedkin-Johnsen (FJ) model~\cite{friedkin1990social}. The DeGroot model assumes that individuals form their opinions by iteratively updating their opinions based on the opinions of their neighbors. The iterations continue until the opinions converge to a steady state. The FJ model is a more nuanced model that takes into account not only the opinions of the neighbors but also the influence of external information sources. In the FJ model, individuals have both confidence and susceptibility parameters. The confidence parameter determines how strongly an individual's own opinion influences their neighbors, while the susceptibility parameter determines how strongly an individual's opinion is influenced by their neighbors. While both models rely on iterative opinion updates, the FJ model introduces additional parameters that allow for the incorporation of external information sources and more nuanced individual opinions.

Several social theories have been applied to the analysis of opinion dynamics, including theories such as the herd behaviour, informational cascade, and homophily. ~\cite{banerjee1992simple},~\cite{bikhchandani1992theory},~\cite{dandekar2013biased}. In~\cite{dandekar2013biased}, they generalizes DeGroot’s model to include biased assimilation. Their results show that homophilous networks lead to polarization only if individuals have a strong bias. Therefore, homophily by itself is not sufficient to polarize society. 

Recently more attention was given to the media and how it affects opinion dynamics and polarization. In~\cite{perra2019modelling} they studied how individuals revise opinions in social networks under algorithmic personalization by introducing a model examining the effects of different filtering algorithms—random, time-based, and semantic filtering—on opinion dynamics among interconnected users.  A research more closely aligned with our paper is~\cite{candogan2022social},  the continuous model in this paper bears resemblance to that of~\cite{candogan2022social} and draws inspiration from it. Our results deviate from theirs in that we offer findings for broad categories of platform functions that are harmonizing and polarizing. We also provide specific examples of functions for which we have pinpointed all potential steady states. 
Also, two recent research includes the media in opinion dynamics in~\cite{auletta2023impact} they investigate how social media recommendations affect opinion dynamics, agents are assigned discrete opinions. Their results show that strong social media influence can lead to polarization, while a high homophily can reduce the impact of social media, which reduces consensus probability. In contrast,~\cite{HU2024128976} focuses on continuous opinions, where agents hold opinions evolve in real numbers. 

While there is a continuing discussion on whether the Internet and social media amplify or mitigate polarization. A study by~\cite{allcott2020welfare} supports the idea that social media platforms, particularly Facebook, may contribute to polarization.  Additionally,~\cite{lelkes2017hostile} shows that internet access may also contribute to an increase in polarization.

Researchers have raised concerns regarding the potential role of social media in driving political polarization such as~\cite{boxell2017greater,levy2021social,flamino2023political,falkenberg2022growing,garcia2023influence,tornberg2021modeling}.
In~\cite{boxell2017greater} they examine polarization using eight different measures to create an index of political polarization among U.S. adults, while in ~\cite{levy2021social}, the study conducted a large field experiment offering subscriptions to either conservative or liberal news outlets on Facebook, examining the impact on different aspects such as news consumption, attitudes, and political opinions. Their results show that exposure to ideologically congruent news influences the slant of visited news sites and decreases negative attitudes towards the opposing party, but does not impact political opinions; However, Facebook's algorithm tends to reduce exposure to opposing views which potentially increasing polarization. Similarly, in~\cite{flamino2023political} they analyze Twitter data to document changes in Twitter’s political news media landscape and measure the resulting polarization induced by social media influencers and their audiences. Also, in~\cite{falkenberg2022growing} they analyzed Twitter data to analyze the polarization over climate change topics. Their results show the importance of analysing the polarisation and its effects on people.

The authors of~\cite{garcia2023influence,tornberg2021modeling}  analyse the polarization in social media. In~\cite{garcia2023influence} they examine the impact of Facebook's news feed algorithms on political polarization,while in~\cite{tornberg2021modeling} they propose a model to understand the emergence of affective polarization in social media. It examines the notions of "echo chambers" and "filter bubbles," discussing the role of social media and group affiliation in human psychology in polarization. In~\cite{azzimonti2018social} they investigated the role of social media network structures and the existence of fake news on misinformation and polarization in society.  One of the findings is how significant misinformation and polarization can arise if only 15\% agents believe fake news to be true. Also, in ~\cite{baumann2020modeling}  they model the dynamics of radicalization and the formation of echo chambers, showing how user engagement and opinions can result in more extreme views. The results emphasize the reasons that may contribute to the existence of echo chambers and polarization in social media networks. 

The papers~\cite{bakshy2015exposure},~\cite{levy2021social}, and~\cite{flaxman2016filter} all support the idea that social media platforms can increase political polarization. In addition, in~\cite{amelkin2017polar} the authors studied how polarized opinions spread through social networks. They used differential equations to describe the dynamics of opinion change over time and proved that opinions converge in the limit using mathematical tools such as Lyapunov function and LaSalle's invariance principle.  In contrast to the aforementioned research, the analysis in~\cite{boxell2017greater}  reveals an opposite result which is that polarization has grown the most in demographic groups that use the internet and social media the least which means the Internet and social media may lead to decreased polarization. 

Although the research community has extensively studied polarization. However, there is still much to learn and many unanswered questions to explore, and ongoing research in this field is essential for developing a deeper understanding of polarization and how to mitigate its negative effects.

\subsection{Outline}
The structure of this paper is as follows: Section~\ref{sec:Model} introduces the model, detailing the classification of platform functions. The results are divided into four parts and presented in Sections~\ref{sec:generalResultsTwo-Agent},\ref{sec:generalResultsCompleteGraph},~\ref{sec:SBM}, and~\ref{sec:specificResults}. In Sections~\ref{sec:generalResultsTwo-Agent} and\ref{sec:generalResultsCompleteGraph}, we examine the outcomes for general platform functions in a two-agent system and complete graphs. In Section~\ref{sec:SBM} we generalise our results to stochastic block model of two blocks, and Section~\ref{sec:specificResults} is dedicated to the results for specific platform functions. The paper concludes with Section~\ref{sec:Conclusion}, where we summarize our findings and discuss potential future work.

\section{Model and Notation}\label{sec:Model}
For $k \in \mathbb{N}$, let $[k]=\{1,2, \dots k \}$.
Let  $V$ be the set of all agents/vertices and $n=|V|$.
Throughout the paper we assume all agents are connected, i.e., in a clique except in Section~\ref{sec:SBM}.  The influence matrix of this graph is represented by a matrix \( A \). The \((i, j)\)-th entry \( a_{ij} \geq 0 \) in \( A \) denotes the influence of agent \( j \)'s on agent \( i \). In this paper, we assume that the matrix \( A \) is symmetric, meaning \( a_{ij} = a_{ji} \).

We use the model proposed by~\cite{candogan2022social}, where agents update their opinion based on their neighbours' opinions and the platform's influence.In this model, given a set of $n$ agents, the opinion of each agent \( i\in [n] \) at time \( t \geq 0 \) is denoted by \( x_i(t) \in [-1, 1] \). 

Here, the domain of \( x_i(t) \) signifies that the opinions can range from strongly negative (-1) to strongly positive (+1), encapsulating all possible opinion states. We assume that we have $n$ agents, and every agent is connected to all the other agents. 
In this paper we consider the system of two agents, which can also be interpreted as two cliques connected (where each agent inside a clique/community has the same initial value), complete graphs, and stochastic block model of two blocks.

To ensure that the opinions do not exceed the bounds of \([-1,1]\), the platform functions are designed to return values within this range only. Each agent’s opinion is influenced by two forces, by their connections and the platform function. Let \( b > 0 \) present the strength of the platform's influence. The opinions are posited to evolve according to \eqref{eq:maineq},and in the case of the two agent system. 
For $i,j\in \{1,2\}$ with $j\neq i$, the differential equation is given by
\begin{equation}\label{eq:difeq}
    \dot{x_i}(t) = a_{i j}(x_j(t) - x_i(t)) + b \cdot (f(x_i) - x_i(t))
\end{equation}
where \( \dot{x_i}(t) \) indicates the rate of change of \( x_i \) over time, and \( f: [-1,1] \rightarrow [-1,1] \) is the function employed by the platform to curate content for the user.

The opinion dynamics presented in  \eqref{eq:maineq} suggests that shifts in an agent's stance are influenced in two ways: the agent's social interactions and the platform's curated content. The first term embodies the continuous social learning akin to the naive (DeGroot-like) model~\cite{degroot1974reaching}, while the second term encapsulates the platform's force, determined by the strength of the platform's influence, \( b \), and the deviation between the agent's current opinion and the platform's targeted content.

We now give the key definitions we require for our results.

\begin{definition}[Equilibrium State or Steady State] 
We say an \opiniondynamic{} $x$ is in equilibrium at time $t^*$ if for all $t \geq t^*$ and all agents $i$, it holds that
$\dot x_i(t) = 0$. A vector or state $(y_1, y_2, \ldots, y_n)$ is an equilibrium, if there is some \opiniondynamic{} that attains its equilibrium in this state. 
\end{definition} 

Note that our definition, as well as the definition of \opiniondynamic, imply that a point $(y_1, y_2, \ldots, y_n)$ is an equilibrium or steady state if an \opiniondynamic{} initiated at this point—specifically, $x_1(0) = y_1$, $x_2(0) = y_2$, \ldots, $x_n(0) = y_n$ remains constant over time.

We divide these equilibria  into two categories:
\begin{definition}[(Strong) Consensus]
The system is in consensus at time $t$ 
if it is in a equilibrium and 
for all $i\in V$, $x_i(t)> 0$,  
for all $i\in V$, $x_i(t)< 0$, or for all $i\in V$, $x_i(t)= 0$.
Again, we call a point a state of consensus if starting at an \opiniondynamic{} in this state, results in consensus for all $t\geq 0$. 
Furthermore, we say the system is in \emph{strong consensus} at time $t$ if it is in consensus at time $t$ and $\forall i, j \in [n], x_i(t) = x_j(t)$, indicating all agents in the graph have the same value.
\end{definition}

Observe that, $(0,0, \ldots,0)$ is always a strong consensus state, due to $f(0)=0$ (which is a consequence of the assumptions we make regarding \properties).

\begin{definition}[Persistent Disagreement]
The system is in persistent disagreement (PD) at time $t$ 
if it is in a equilibrium and 
there exists $i,j\in V$ such that $x_i(t) > 0$ and $x_j(t) < 0$. We will sometimes refer to such states as being polarized.
We say a point is in persistent disagreement if starting an \opiniondynamic{} at it results in a state of persistent disagreement for all $t\geq 0$. 
\end{definition}

As we  will see later, the existence of consensus states and persistent disagreement states are determined by the platform influence function.
We explore a spectrum of platform influence functions, which we classify into two distinct categories based on their dynamic impact on the system:

\begin{definition}[Harmonizing Platform Influence Functions-HPIF]\label{def:harmonizing}
We say a platform influence function $f(\cdot)$ is harmonizing,  if applying this function over time leads to a state where there is no persistent disagreement state.
\end{definition}

\begin{definition}[Polarizing Platform Influence Functions-PPIF]\label{def:polarizing}
We say a platform influence function $f(\cdot)$ is polarizing if applying this function over time leads to at least one persistent disagreement state.
\end{definition}

\subsection{\properties}\label{sec:properties}
We will assume that all platform influence functions adhere to the following properties:

\begin{enumerate}
    \item \textbf{Continuity:} Each function is continuous and Lipschitz continuous, ensuring a smooth transition of influence across the range of values.
    \item \textbf{Symmetry:} They exhibit symmetric behavior about the origin, governed by the property \( f(-x) = -f(x) \).
    \item \textbf{Sign Preservation:} The functions return  a positive value for positive inputs, and a negative value for negative inputs.
    \item \textbf{Boundedness:} To ensure that opinions remain within the range of possible bias \([-1, 1]\), we require that $f(\cdot) \in [-1,1]$.
\end{enumerate}

As we denote a function fulfilling properties 1 to 4 a \emph{feasible function}.
Examples of this family include, but are not limited to: the 
$sgn_\epsilon$ function, see \eqref{eq:signFunction}, and linear functions $f(x)=\alpha x$ for $\alpha <1$ (see \eqref{eq:LinearHarmonizingFunction}).

\subsection{Polarization Planes in Two Agent System}
\label{sec:Polarization_planes_2agents}
It is well-known\cite{pouso2012peano} that for continuous $f$ a solution to the differential equation exists, and according to Theorem 3 in\cite{candogan2022social} there exists an equilibrium that agents converge to it in the limit if the influence matrix $A$ is symmetric. To answer the question of whether the equilibrium is a state of consensus or persistent disagreement in the case of two agent system, we introduce the following concept 
of polarization planes. We define a polarization plane to be a map from $ [-1,1] \times [-1,1] \times (-\infty, \infty) \rightarrow ( -\infty, \infty)$ and then you say, let $G$ and $H$ be two polarization planes to be \emph{polarization planes}, which we will use to delineate critical thresholds in the opinion dynamics model.
As we will see later, they define equilibrium states where the gradient of the system's potential is zero.
For simplicity in explaining the planes we assume that \( a_{ij} = a_{ji}=1 \). 

Based on the ordinary differential equation presented in  \eqref{eq:difeq},
The system will be in an equilibrium if and only if for all $i,j \in [2]$ with $j\neq i$, \( f(x_i(t)) = \frac{x_i(t) - x_j(t)}{b} + x_i(t) \).
We define the polarization planes, to include all points $(x,y)$ that satisfy $f(x)=\frac{x-y}{b} + x$ and $f(y)=\frac{y-x}{b}+y$. 
These points can be represented as three-dimensional planes, namely \( (x, y, \frac{x-y}{b} + x) \) and \( (x, y, \frac{y-x}{b} + y) \). 
See Figure~\ref{fig:SignFunctionWithIntersectionPoints} for an illustration.
We define,
\( G(x,y,b) = \frac{x - y}{b} + x \text{\phantom{x}  and \phantom{x} } H(x,y,b) = \frac{y - x}{b} + y . \)
Similarly, we can define the plane of values to be \(F_1 = \{ (x,y,f(x))~|~ x,y \in [-1,1] \} \text{ and}\)
\(F_2 = \{ (x,y,f(y))~|~ x,y \in [-1,1]\}. \phantom{and}\)The following lemma provides some key insights into the polarization planes.

\begin{lemma}\label{lem:plane}
Consider the system at time $t$.
Let $z=x_1(t)$ and $y=x_2(t)$.
If $(z,y)$ is an equilibrium, it must be that  $F_1$ intersects with $G$ at $(z,y)$ and, simultaneously, $F_2$ intersects with $H$ at $(z,y)$. Moreover, we have that if $f(z) > G(z,y,b)$, then $\dot x_1(t) > 0$. 
Conversely, if  $f(z) < G(z,y,b)$, then $\dot x_1(t) < 0$. Analogue statements hold for $f(y)$ and $H$. 
\end{lemma}
\begin{proof}
Let $\epsilon  = f(z)- G(z,y,b)$.
We have,
$\dot x_1(t) = a (y-z) + b(f(z)-z)= a (y-z) + b(G(z,y,b)+\epsilon-z)=a (y-z) + b(\frac{a (z-y) }{b}+z+\epsilon-z)=b \epsilon$. Thus, only for $\epsilon=0$ the derivative is zero, for $\epsilon>0$ the derivative is positive, etc.
\end{proof}
See Figures~\ref{fig:SignFunctionWithIntersectionPoints}, ~\ref{fig:SignFunctionWithIntersectionPoints2}, ~\ref{fig:SignFunctionWithIntersectionPoints3}, ~\ref{fig:SignFuncIntersection-1}, ~\ref{fig:SignFuncIntersection-2},
~\ref{fig:LinearHarmonizingFunction}, and ~\ref{fig:LinearHarmonizingFunction2} for  illustration. 
Since persistent disagreement can only occur in the quadrants $(+, -)$ (meaning that $x_1(t)>0$ and $x_2(t)<0$) and $(-, +)$, it is sufficient for us to focus on $(+, -)$, due to symmetry: The points in quadrant $(+, -)$ can be mapped to points in quadrant $(-, +)$ by reflecting them across the origin.

The underlying symmetry also allows us to prove that the rotational matrix  \[R=
\begin{pmatrix}
0 & -1 & 0 \\
-1 & 0 & 0 \\
0 & 0 & -1
\end{pmatrix}\]  defines a bijection between $F_1$ and $F_2$ and, at the same time,  between $G$ and $H$. See Figure~\ref{fig:SignFunctionWithIntersectionPoints}, ~\ref{fig:SignFunctionWithIntersectionPoints2}, ~\ref{fig:SignFunctionWithIntersectionPoints3}, ~\ref{fig:SignFuncIntersection-1}, ~\ref{fig:SignFuncIntersection-2},
~\ref{fig:LinearHarmonizingFunction}, and ~\ref{fig:LinearHarmonizingFunction2} for an illustration. 

\begin{figure}[ht]
    \centering
    \begin{minipage}[b]{0.32\textwidth} 
        \includegraphics[width=\textwidth]{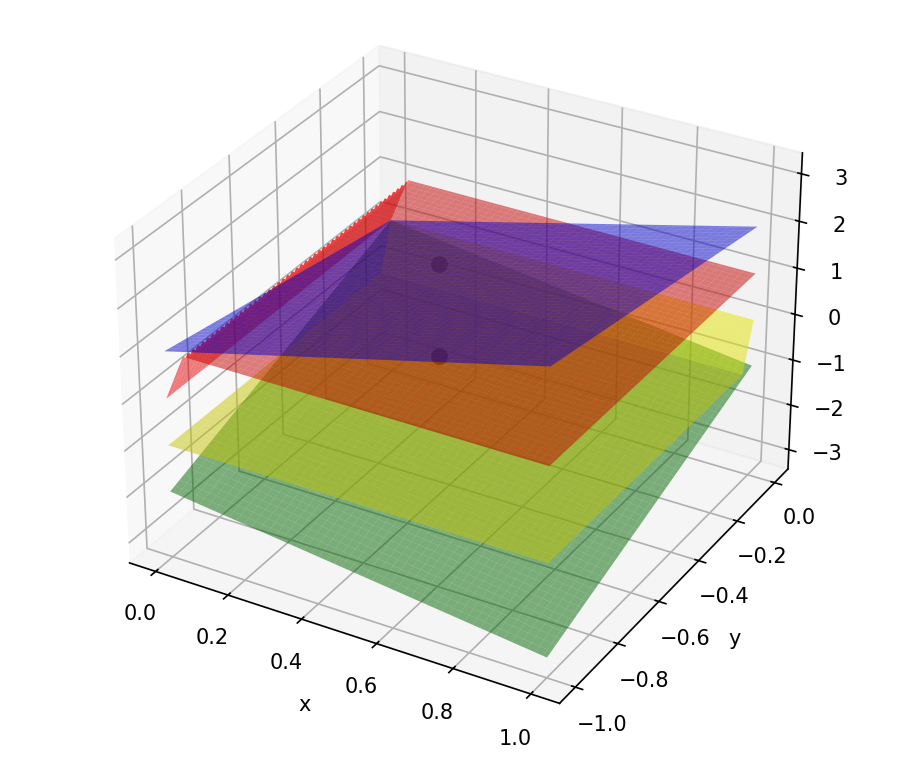}
        \caption{A depiction of the sign function along with the intersection points and polarization planes: \(G\) (purple), \(F_1\) (red), \(F_2\) (yellow) and \(H\) (green).}
        \label{fig:SignFunctionWithIntersectionPoints}
    \end{minipage}
    \hfill 
    \begin{minipage}[b]{0.32\textwidth}
        \includegraphics[width=\textwidth]{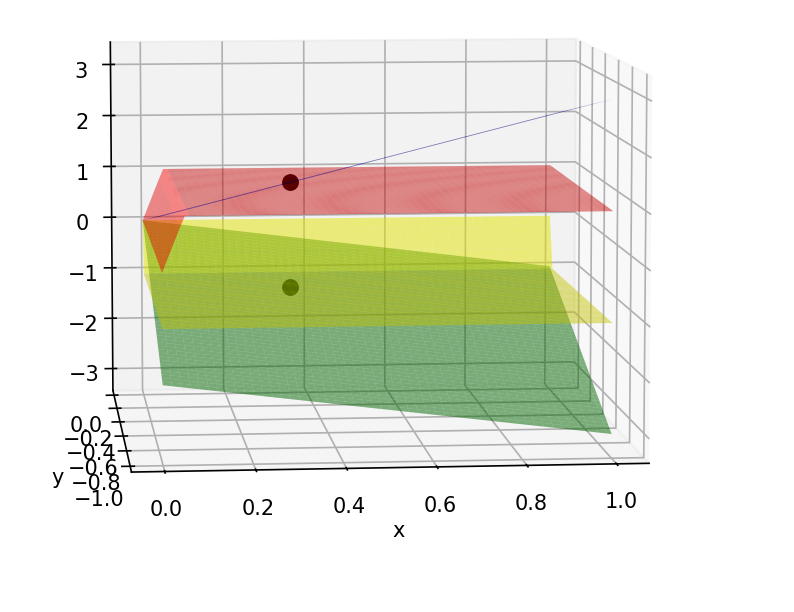}
        \caption{A second viewing angle of Figure~\ref{fig:SignFunctionWithIntersectionPoints} highlighting the intersection point of \(F_1\) (red) and the polarization plane: \(G\) (purple).}
        \label{fig:SignFunctionWithIntersectionPoints2}
    \end{minipage}
    \hfill 
    \begin{minipage}[b]{0.32\textwidth}
        \includegraphics[width=\textwidth]{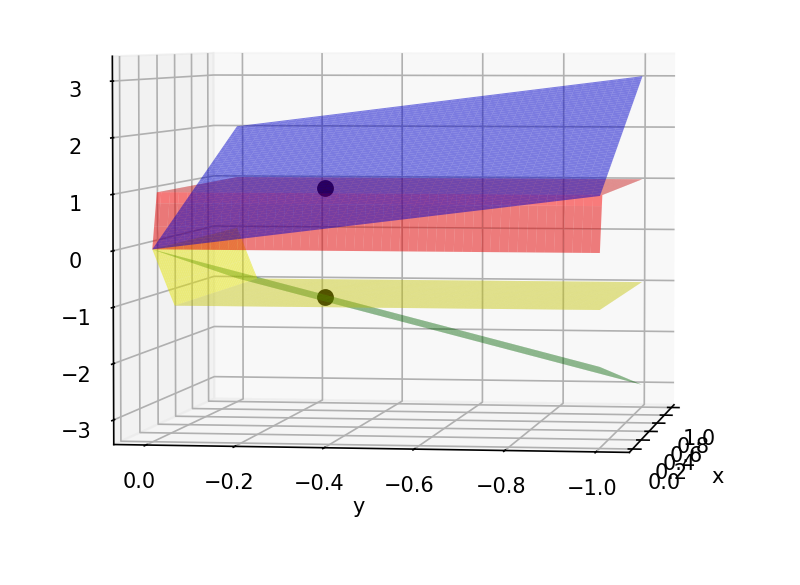}
        \caption{A third viewing angle of Figure~\ref{fig:SignFunctionWithIntersectionPoints} highlighting the intersection point of \(F_2\) (yellow) and the polarization plane: \(H\) (green).}
        \label{fig:SignFunctionWithIntersectionPoints3}
    \end{minipage}
\end{figure}

 \subsubsection{Persistent Disagreement Results}
Using the polarization planes and platform function planes we can summarize our results for persistent disagreement into the following three cases.
Due to the symmetry mentioned before, we will only focus on the octant where $x>0, y<0$, and $f(x)\geq 0$. 
\begin{itemize}
    \item The plane $F_1$ intersects the $G$ plane at least at one point, i.e., $F_1 \cap G \neq \emptyset$.
    Due to symmetry, $F_2$ must also intersect the $H$ plane. Then,  there exists at least one persistent disagreement state. See Theorem~\ref{thm:B}.
    \item The entire plane $F_1$ is located below the $G$ plane and, due to its symmetry, $F_2$ is above $H$.     In this case, there is no persistent disagreement state see Theorem~\ref{thm:C}. 
    \item The entire plane $F_1$ is located above the $G$, i.e., for all $x >0 $ and $y<0$, $f(x) \geq \frac{x-y}{b} + x$. Due to symmetry, $F_2$ must be completely below the $H$ plane. 
    This cannot occur, because there is no $F_1$ completely above $G$ since the platform function is, by definition, bounded by $1$ while the plane $G$ exceeds 1.
\end{itemize}
For item 1, we call this function the polarizing function as we will elaborate in Section \ref{sec:PPIF}. For item 2, we call this function the harmonizing function as we will elaborate in Section \ref{sec:HPIF}.
\subsection{Polarization Conditions in Complete Graphs}
\label{sec:Polarization_planes_Nagents}
The concept of polarization planes can be extended for complete graphs with \( n \) vertices which we denote as \( K_n \). Since we have complete graph, then each agent $i$ has a fixed number of connections $n-1$ connections. Therefore, the weight of each connection $\frac{a}{n-1}$, i.e., $a_{ij} = \frac{a}{n-1}$ is the influence of agent $j$ on agent $i$. Hence, the opinions of agent $i$ evolve according to the equation
\[\dot{x}_{i}(t)=\frac{a}{n-1} \sum_{j \in N(i)}\left(x_{j}(t)-x_{i}(t)\right)+b\left(f\left(x_{i}(t)\right)-x_{i}(t)\right)
\]
where $n$ is the total number of agents in the network, $a >0$ is the influence parameter. In the equilibrium where \(\dot{x}_{i}(t) =0\) we have \( f(x_i) = \frac{(a + b)x_i(t) - a \bar{x}_i(t)}{b}\). 
where  \( \bar{x}_i(t) \) is the mean of all the neighboring nodes of node \( i \) in the complete graph (excluding \( i \)), it can be written as:
\(
\bar{x}_i(t) = \frac{1}{n-1} \sum_{\substack{j \in N(i)}} x_j(t)
\)
. To be consistent with the two agent system lets call this function G such that:
\[G(x_i(t),\bar{x_i}(t),a,b) = \frac{(a + b)x_i(t) - a \bar{x}_i(t)}{b}
\label{eq:G_Nagents}
\] If there exists a time $t^* \geq 0$ such that $G(x_i(t^*),\bar{x_i}(t^*),a, b) = f(x_i(t^*))$, for every agent $i$ in the state vector \(\mathbf{x} = (x_1(t), x_2(t), \ldots, x_n(t))\) of all agents in the graph, then all agents in $K_n$ have reached equilibrium.
We can define the plane of values as
\(F_i = \{ (x_i(t), \bar{x_i}(t), f(x_i(t))) \mid x_i(t), \bar{x_i}(t) \in [-1,1] \},\) 
The lemma \ref{lem:planes_complete} provides some key insights into the polarization planes in complete graphs.

\section{Results for General Platform Functions for Two-Agent System}
\label{sec:generalResultsTwo-Agent}
In this section we present our results of characterizing necessary and sufficient conditions for strong consensus states and persistent disagreement states in the case of the two-agent system.
\begin{theorem}[Sufficient and Necessary Conditions for Strong Consensus]\label{thm:A}
Fix arbitrary $b >0$ and $a >0$, and consider an \opiniondynamic{} with two agents $i$ and $j$. Assume $f(\cdot)$ fulfills the \properties{}. Then,  for any strong consensus state $(y,y)$ it is sufficient and necessary that $f(y) = y$. 
\end{theorem}
\begin{proof}
By symmetry properties from $f(y)=y$ we obtain $f(-y)=-y$. Moreover, recall that $f(0)=0$ due to the symmetry assumption. Let $\bar y \in \{-y,0,y\}$ and consider an \opiniondynamic{} $x$ starting in $(\bar y,\bar y)$. Then $ x_i(0)=x_j(0)=f(x_i(0))=f(x_j(0)))=\bar y$ by definition. 
Thus, $\dot x_i(0)= a (\bar y -\bar y)  + b(\bar y - \bar y) =0$ which is analogue for $j$. This shows that all the points $(y,y)$ $(-y,-y)$ and $(0,0)$ are equilibrium states. By definition, they are also consensus states.
\end{proof}
Note that this characterization is only for states of strong consensus.
There may be other consensus states depending on the platform function, for example, $(\frac12, \frac34)$ if $f(\frac12)=\frac14$, $f(\frac34)=1$, and $b =1$. However, a necessary and sufficient condition for consensus to occur is naturally that $F_1$ and $G$ intersect at the same $x$ and $y$ coordinates as $F_2$ 
intersect with $H$.
Observe that Theorem implies that there is at least one consensus state in the quadrant $(+,+)$ and one in the quadrant $(-,-)$.
\begin{theorem}[Sufficient Condition for PD] \label{thm:B}
Fix arbitrary $b >0$ and $a >0$, and consider an \opiniondynamic{} $x$ with two agents $i$ and $j$. Assume $f(\cdot)$ fulfills the \properties. If there exist $z \in (0,1]$ and $y \in [-1,0)$ such that \( f(z) \geq G(z, y, a,b) \) and \( f(y) \leq H(z, y,a, b) \), then there exists at least one persistent disagreement state. 
\end{theorem}

\begin{proof}
Consider $x_i(0)=z$ and $x_j(0)=y$.
W.l.o.g. assume that \( f(z) > G(z, y,a, b) \) and \( f(y) < H(z, y,a, b) \).
Thus, assume $f(x_i(0))> G(x_i(0),x_j(0),a,b)$.  Define
\begin{align}
\label{eq:initialdiscrepancies}
    \epsilon(t) &:= f(x_i(t)) - G(x_i(t), x_j(t),a, b).
\end{align}
At $t=0$ we get from our assumption that $\epsilon(0)>0$.
The derivatives \( \dot{x_i}(t) \) and \( \dot{x_j}(t) \) can be described in terms of $\epsilon(t)$ which can be seen by applying the definition of the functions $\epsilon$ and $G$. More concrete,
\begin{align}
\label{eq:derivatives}
    \dot{x_i}(t) &= a(x_j(t) - x_i(t)) + b\left( \frac{a (x_i(t)-x_j(t))}{b} +x_i(t) + \epsilon(t) - x_i(t)\right) = b \cdot \epsilon(t), 
\end{align}
Analogously, it follows that, $\dot{x_j}(t) = -b \cdot \epsilon(t)$.
This means that \( x_i(t) \) is increasing as long as $\epsilon(t) > 0$ and \( x_j(t) \) is decreasing for the same amount of time and at the same rate. 

In other words, at time $0$, $(x_i(0), x_j(0))$ lies above $G$, and $x_i(t)$ is increasing until there exists an intersection, i.e., a point $(x^*,-x^*)$ such that $F_1$ intersects with $G$. By the boundedness of $f(\cdot)$ and since $b>0$, we have  $G(x_i(t),x_j(t),a, b) > 1$.
Thus, since, $f(\cdot)$ is continuous, we have, by the intermediate value Theorem, that there exists an $(x^*,-x^*)$ such that $F_1$ intersects with $G$ due to the symmetry. By Symmetry $F_2$ at this point also intersects with $H$. Thus, by definition, $G$ and $H$ the resulting point must be an equilibrium point and due to the different signs, it is a persistent disagreement state. 
\end{proof}
\begin{theorem}[Necessary Condition for PD]\label{thm:C}
Fix an arbitrary $b >0$ and $a >0$. Consider an \opiniondynamic{} $x$ with two agents $i$ and $j$..
Assume $f(\cdot)$ fulfills the \properties.
If a persistent disagreement state exists at $x_i(t)$ and $x_j(t)$, then
there must exist  $z \in (0,1]$ and $y \in [-1,0)$ such that \( f(z) \geq G(z, y,a, b) \) and \( f(y) \leq H(z, y,a, b) \).
\end{theorem}
\begin{proof}
We will prove the equivalent statement:
If for all $x_i(0) \in (0,1]$ and $x_j(0) \in [-1,0)$ we have that \( f(x_i(0)) < G(x_i(0), x_j(0),a, b) \) and \( f(x_j(0)) > H(x_i(0), x_j(0), a, b) \), then no persistent disagreement steady state exists, and both agents will reach consensus.

To demonstrate the absence of an equilibrium state in the $(+,-)$ quadrant, we analyze the system's dynamics under the given conditions. For \( x_i(0) > 0 \) and \( x_j(0) < 0 \), \( f(x_i(0)) < G(x_i(0), x_j(0),a, b) \) implies, by Lemma~\ref{lem:plane}, a negative derivative \( \dot{x_i}(t) \), leading to a decrease in \( x_i(t) \). Conversely, \( f(x_j(0)) > H(x_i(0), x_j(0), a, b) \) yields a positive derivative \( \dot{x_j}(t) \) (by Lemma~\ref{lem:plane}), causing \( x_j(t) \) to increase. These dynamics ensure that \( f(x_i(t)) \) never equals \( G(x_i(t), x_j(t), a, b) \) and \( f(x_j(t)) \) never equals \( H(x_i(t), x_j(t), a, b) \), precluding an equilibrium state within this quadrant.

Since we do not have persistent disagreement states, the only equilibrium states are consensus steady states, and given the continuity of the function \( f(\cdot) \) in equation ~\ref{eq:difeq}  and the influence matrix $A$ is symmetric which is 1s in our case and by applying Theorem 3 from\cite{candogan2022social}, we guarantee that agents’ opinions converge to equilibrium and since we do not have persistent disagreement states the system will converge to consensus steady states.
\end{proof}
 
As an immediate consequence of our methods, we can see that the degree to which agents disagree in persistent disagreement states is a function of the platform influence $b$: all those PD states must lie on the polarization plane which is a function of $b$. 

\section{Results for General Platform Functions for Complete Graph of n Agents}
\label{sec:generalResultsCompleteGraph}
In this section we present the necessary and sufficient conditions for strong consensus and persistent disagreement states in complete graphs consisting of n agents.

\begin{lemma}
\label{lem:planes_complete}
Consider the system at time $t$.
Let $z_i=x_i(t)$, if $(z_1, z_2, z_3, \ldots, z_n)$ is an equilibrium, it must be that \( F_i \) intersects with \( G \) at \( (z_i, \bar{x_i}(t),a, b) \) for all \( i \) agents. Moreover, we have that if $z_i>0$ and $f(z_i) > G(z_i,\bar{x_i}(t),a,b)$, then $\dot x_i(t) > 0$. 
Conversely, if $f(z_i) < G(z_i,\bar{x_i}(t),a, b)$, then $\dot x_i(t) < 0$. Analogue statements hold if $z_i<0$.
\end{lemma}
\begin{proof}
Let $\epsilon = f(z_i) - G(z_i,\bar{x}_i (t),a, b)$. We have,
\[
\dot{x}_{i}(t) = \frac{a}{n-1} \sum_{j \in N(i)} \left( x_j(t) - z_i \right) + b \left( f(z_i) - z_i \right)
\]
\[
\dot{x}_{i}(t) = \frac{a}{n-1} \sum_{j \in N(i)} \left( x_j(t) - z_i \right) + b \left( \frac{(a + b)z_i - a \bar{x}_i(t)}{b} + \epsilon - z_i \right)
\]
\[
\frac{(a + b)z_i - a \bar{x}_i(t)}{b} + \epsilon - z_i = \frac{(a + b)}{b} z_i - \frac{a \bar{x}_i(t)}{b} + \epsilon - z_i
\]
\[
= \left( \frac{a + b}{b} - 1 \right) z_i - \frac{a \bar{x}_i(t)}{b} + \epsilon = \frac{a}{b} z_i - \frac{a \bar{x}_i(t)}{b} + \epsilon
\]
\[
\dot{x}_i(t) = \frac{a}{n-1} \sum_{j \in N(i)} \left( x_j(t) - z_i(t) \right) + b \left( \frac{a}{b} z_i(t) - \frac{a \bar{x}_i(t)}{b} + \epsilon \right)
\]
\[
\dot{x}_i(t) = \frac{a}{n-1} \sum_{j \in N(i)} \left( x_j(t) - z_i(t) \right) + a z_i(t) - a \bar{x}_i(t) + b \epsilon
\]
\[
\frac{a}{n-1} \sum_{j \in N(i)} \left( x_j(t) - z_i(t) \right) = a (\bar{x}_i(t) - z_i(t))
\]
\[
\dot{x}_i(t)  = a (\bar{x}_i(t) - z_i(t)) + a z_i(t) - a \bar{x}_i(t) + b \epsilon
\]
\[
\dot{x}_i(t)  = a \bar{x}_i(t) - a z_i(t) + a z_i(t) - a \bar{x}_i(t) + b \epsilon
\]
\[
\dot{x}_i(t)  = b \epsilon
\]
Thus, only for $\epsilon = 0$ the derivative is zero, for $\epsilon > 0$ the derivative is positive, etc.
\end{proof}
\begin{theorem}[Strong Consensus]\label{thm:AN}
Fix arbitrary $b > 0$ and $a >0$, and consider a complete graph with $n$ vertices, $K_n$. Assume $f(\cdot)$ fulfills the \properties{}. Then,  for any strong consensus state $
(y, y, \dots, y) \in \mathbb{R}^n
$ it is sufficient and necessary that $f(y) = y$ where $y\in [-1, 1]$. 
\end{theorem}
\begin{proof}
Proof by contradiction. Assume \[ \exists y \in [-1, 1] \, \text{such that} \, f(y) \neq y \] and $x=(y, y, \dots, y)$ is an equilibrium state. According to the differential equation
\[
\dot{x}_i(t) = \frac{a}{n-1} \sum_{j \in N(i)} \left( y - y \right) + b \cdot (f(y) - y) = 0.
\]
Since the term \( \sum_{j \in N(i)} (y - y) = 0 \), this simplifies to 
\[
\dot{x}_i(t) = b \cdot (f(y) - y) =0.
\]
For this to hold, it must be that \( f(y) = y \). This contradicts the assumption that \( f(y) \neq y \), and therefore, we conclude that \( f(y) = y \) is both necessary and sufficient for consensus.
\end{proof}
Note that this characterization is only for states of strong consensus.
There may be other consensus states depending on the platform function. However, a necessary and sufficient condition for consensus to occur is naturally that the platform influence function intersects with $G$ for all agents of the graph in the same time $t$.\\
\begin{theorem}[Necessary Condition for PD]\label{thm:BN}
Fix an arbitrary $b > 0$ and $a > 0$ and consider a complete graph \( K_n \) with \( n \) agents, where each agent \( i \) has an opinion \( x_i(t) \) at time \( t \). Define:
\begin{itemize}
    \item \( X^+(t) = \{ x_i(t) \mid x_i(t) \geq 0 \} \) to be the set of states of agents with non-negative opinions.
    \item \( X^-(t) = \{ x_j(t) \mid x_j(t) \leq 0 \} \) to be the set of states of agents with non-positive opinions.
\end{itemize}
Assume \(f(\cdot)\) fulfills the \properties{}, and where \( \bar{x}_i(t) \) is the mean of all the neighboring nodes of node \( i \) . If a persistent disagreement state exists then there must exist at least one agent in \(X^+(t)\) and \(X^-(t)\), such that 
\begin{itemize}
    \item For all \( x_i(t) \in X^+(t) \), \( f(x_i(t)) \geq G(x_i(t),\bar{x_i}(t),a,b) \)
    \item For all \( x_j(t) \in X^-(t) \), \( f(x_j(t)) \leq G(x_j(t), \bar{x_i}(t),a,b) \)
\end{itemize}
\end{theorem}
\begin{proof}
We will prove the equivalent statement: If for all \(x_i(t) \in X^+(t) \) and \( x_j(t) \in X^-(t) \) we have that \( f(x_i(t))<G(x_i(t),\bar{x_i}(t),a, b) \) and \(f(x_j(t))>G(x_j(t), \bar{x_i}(t),a, b)\), then no persistent disagreement steady state exists, and all agents will reach consensus.
Consider a system with \( k^+ \) positive agents and \( k^- \) negative agents in \( \mathbb{R}^n \). Assume we begin with an initial vector \( \mathbf{x}(0) \) that contains values from the sets \( X^+(t) \) and \( X^-(t) \), representing positive and negative agents respectively, and there exists an equilibrium \[
\mathbf{x}^* = 
\begin{bmatrix}
\mathbf{x}^{+*}_1 & \cdots & \mathbf{x}^{+*}_{k^+} & \mathbf{x}^{-*}_1 & \cdots & \mathbf{x}^{-*}_{k^-}
\end{bmatrix}
\]
where \(\mathbf{x}^*\) represent a Persistent Disagreement Equilibrium (PDE) such that \( \dot{\mathbf{x}}^*(t) = 0 \) and \[
\forall i, \; (\mathbf{x}_i(0) > 0 \implies \mathbf{x}_i(0) < \mathbf{x}_i^* \; \wedge \; (\mathbf{x}_i(0) < 0 \implies \mathbf{x}_i(0) > \mathbf{x}_i^*)
\] Now, consider the dynamical behavior of agents in \( \mathbf{x}(0) \):
\begin{itemize}
  \item For any \( x_i(t) \in X^+(t) \), since \( f(x_i(t)) < G(x_i(t), \bar{x_i}(t),a, b) \), it leads to \( \dot{x}_i(t) < 0 \), indicating that \( x_i \) is decreasing.
  \item Conversely, for any \( x_j(t) \in X^-(t) \), since \( f(x_j(t)) > G(x_j(t), \bar{x_i}(t),a, b) \), then \( \dot{x}_j(t) > 0 \), implying that \( x_j \) is increasing.
\end{itemize}
Given these conditions, the dynamics \( \dot{x}_i(t) \) and \( \dot{x}_j(t) \) suggest that the system's state is moving away from the (PDE) state \( \mathbf{x}^* \). If these inequalities hold consistently for all agents, then no Persistent Disagreement Equilibrium (PDE) can exist in the long term, as the dynamics drive the system away from the state \( \mathbf{x}^* \). Since we do not have persistent disagreement states, and by applying Theorem 3 from~\cite{candogan2022social}, the system will converge and the only equilibrium states are consensus steady states.
\end{proof}
\section{Results for the Stochastic Block Model with Two Blocks}\label{sec:SBM}
In this section we generalize our results in two-agent system to the stochastic block model (SBM) with two blocks. We show that the equilibrium of two-agent system can be approximated when the network is large with a lot of connections. This section and the main associated Theorem is inspired by Theorem 2 in~\cite{candogan2022social} we instead of proving it for the steady states of sign function, we try to generalize it for any function that adheres the platform function properties. We consider two sets of agents, left and right, with connection probabilities defined within and between these sets. The initial opinions of each block are random variables assigned from predefined left/ right sets. Since in the proof we will use the mean field approximation, we have to normalize the influence matrix $A$ to a single scalar $a$. We assume that if agent \( i \) has \( d_{i} \) different connections, each connection carries a weight of \( a / d_{i} \). Specifically, \( a_{i j} = a / d_{i} \) represents the influence of agent \( j \) on agent \( i \), provided \( j \) is a neighbor of \( i \); otherwise, \( a_{i j} = 0 \). The opinion of agent \( i \) evolves according to the following equation.
\[\dot{x}_{i, n}(t)=\frac{a}{|N(i)|} \sum_{j \in N(i)}\left(x_{j, n}(t)-x_{i, n}(t)\right)+b\left(f\left(x_{i, n}(t)\right)-x_{i, n}(t)\right)
\]
Where \(N(i)\) represents the set of connections for agent \(i\), and \(|N(i)|\) denotes the total number of connections that agent \(i\) has, we build on similar notations from previous work~\cite{candogan2022social}. Theorem~\ref{thm:SBM} demonstrates that in dense networks, as the number of agents increases, the opinions of left (right) agents converge over time to the equilibrium opinion of the corresponding left (right) agent in a simplified two-agent system, provided certain conditions on the initial opinion distribution are met. 

We define \(A_n\) as an adjacency matrix, where \(A_{n,ij} = 1\) indicates a connection between agent \(i\) and agent \(j\), and \(A_{n,ij} = 0\) otherwise. The set of all possible adjacency matrices for a network of \(n\) left and right agents is denoted by \(\boldsymbol{A}_n\). The stochastic block model (SBM) with two blocks generates a probability space \(\left( \boldsymbol{A}_n, 2^{\boldsymbol{A}_n}, \mathbb{P} \right)\), where \(2^{\boldsymbol{A}_n}\) is the set of all subsets of \(\boldsymbol{A}_n\), and \(\mathbb{E}\) is the expectation operator corresponding to the probability measure \(\mathbb{P}\).

The initial opinion of each left agent is drawn from a random variable on \(X_{L} \subset \mathbb{R}_{-} = (-\infty, 0]\), and for each right agent from \(X_{R} \subset \mathbb{R}_{+} = [0, \infty)\). The distribution of initial opinions is captured by the probability measure \(\mu_0\), which is supported on the finite set \(\boldsymbol{X}_0^n\) in \(\mathbb{R}^{2n}\), where \(\boldsymbol{X}_0 = X_L \times X_R\). The SBM with two blocks defines a product probability space \(\left(\boldsymbol{A}_n \times \boldsymbol{X}_0^n, 2^{\boldsymbol{A}_n} \times 2^{\boldsymbol{X}_0^n}, \mathbb{P} \otimes \mu_0\right)\), where \(\otimes\) represents the product of \(\mathbb{P}\) and \(\mu_0\). 

Given a realized adjacency matrix \(A\), the opinion of each agent evolves dynamically based on their connections and platform influence. Following~\cite{candogan2022social}, a function \(f(n)\) is said to be \(\omega(g(n))\) if for every \(c > 0\), there exists \(n_0 \geq 1\) such that \(f(n) > c g(n)\) for all \(n \geq n_0\). We define \(X_{eq}(a, b)\) as the set in \(\mathbb{R}^2\) that includes initial opinions that converge to equilibrium under parameters \(a\) and \(b > 0\). The opinion of agent \(i\) at time \(t\), given the adjacency matrix \(A_n\) and initial opinions \(\boldsymbol{x}_0\), is denoted by \(x_{i, n, A_n, \boldsymbol{x}_0}(t)\).

\begin{theorem}\label{thm:SBM}
Assume that \( p(n) \) and \( q(n) \) satisfy \( \omega\left(\frac{\ln n}{n}\right) \). Let \( \beta = \lim_{n \to \infty} \frac{q(n)}{q(n) + p(n)} \). Consider a two-agent system characterized by \( (a\beta, b, (x_{L}, x_{R})) \) with an equilibrium \( \boldsymbol{e} = (e_{L}, e_{R}) \), and $f(\cdot)$ fulfills the Platform Influence Functions Properties and it is monotonically increasing. If \( X_{L} \times X_{R} \subseteq X_{eq}(a\beta, b) \) where \( X_{eq}(a\beta, b) \) is the set of initial values that lead to this equilibrium, then for every \( \epsilon > 0 \), there exist \( \delta > 0 \), \( N > 0 \), \( T > 0 \), and sets \( C_{\delta, n} \subset \boldsymbol{A}_{n} \) such that \( \lim_{n \rightarrow \infty} \mathbb{P}(C_{\delta, n}) = 1 \) and for all \( n \geq N \), all \( t \geq T \), all \( A_{n} \in C_{\delta, n} \), all \( \boldsymbol{x}_{0} \in (X_{L} \times X_{R})^{n} \), 
we have
\begin{equation}
\label{eq:SBM}
|x_{i, n, A_{n}, \boldsymbol{x}_{0}}(t) - e_{L}| \leq \epsilon \text{ and } |x_{j, n, A_{n}, \boldsymbol{x}_{0}}(t) - e_{R}| \leq \epsilon 
\end{equation}
for every left and right agent \( i \). \end{theorem} Then, Inequality \ref{eq:SBM} holds when the network is large and dense. So that the influence of an agent on her neighbours is tiny with high probability. In other words, this Theorem is applying the mean-field approximation for the opinion dynamics that are generated by the SBM. This Theorem shows that if we have an equilibrium in two agents system, then if the network is dense enough and based on initial opinions, and the platform influence function satisfies the platform properties then the agents' opinions in this general stochastic block model can be approximated to the opinion of two agent system. 
The expected number of connections from right to left or left to right is denoted as \( q(n)n \),The expected total number of connections is \( q(n)n + p(n)n \). When the network is sufficiently dense, meaning \( p(n) \) and \( q(n) \) satisfy \( \omega\left(\frac{\ln(n)}{n}\right) \) and \( n \) is large, these expected values hold. Then, by applying Lemma 4 from~\cite{candogan2022social}, which employs a concentration argument, we see that with high probability, each agent has approximately \( q(n)n \) connections of the opposite type, with a small margin of error, and around \( p(n)n + q(n)n \) total connections, again with a small error. Then, we construct differential equations to establish bounds on the paths of all left and right agents, each of whom has approximately \( q(n) n \) connections to agents of the opposite type and a total of \( q(n) n + p(n) n \) connections. These bounds show that for large \( t \), the opinions of all agents in both blocks are close to the equilibrium opinion of a single representative agent in two agent system. Note that, the opinion of each left or right agent is almost not influenced by the other left (right) agents and is approximately influenced by the single right (left) representative agent at a rate of  \( a q(n) n /(p(n) n+q(n) n) \approx a \beta \).
The influence parameter is $\beta$ since \( p(n) n + q(n) n \) is the expected total number of connections and \( q(n) n \) is the expected number of connections with the other block. Therefore, the equilibrium states of the SBM can be approximated by those of a two-agent system for large \( n \) and \( t \). 
The details of how we proved that can be seen in Appendix~\ref{sec:appendix-SBM}.  
\section{Results for Specific Platform Functions}
\label{sec:specificResults}
In this section, we study explicit platform signaling functions and consider both polarizing and harmonizing functions.
\subsection{Examples of Polarizing Platform Influence Functions (PPIF)}
\label{sec:PPIF}
One example of a polarizing platform influence function is the sign function - to be more precise, we will study a continuous version of it to remain within our framework. 
\subsubsection{Sign Function}
This function was used by~\cite{candogan2022social}
as their platform function. This platform function pushes each agent's opinion towards content that aligns with their current view. It is a polarizing function as it intersects with the polarization planes as shown in the Figures~\ref{fig:SignFunctionWithIntersectionPoints}, ~\ref{fig:SignFunctionWithIntersectionPoints2}, ~\ref{fig:SignFunctionWithIntersectionPoints3}, \ref{fig:SignFuncIntersection-1} and \ref{fig:SignFuncIntersection-2} for the $(+,-)$ quadrant. 
The purple plane represents $G$, the green plane represents $H$, the red plane represents \((x_i(t),x_j(t),f(x_i(t)))\), and the yellow plane represents\((x_i(t),x_j(t), f(x_j(t)))\) with \(\epsilon = 0.05\), where \(\epsilon\) is a threshold. 

In~\cite{candogan2022social} the authors present a continuous approximation of the sign function for a very small $\epsilon \textgreater 0$
\[    \label{eq:signFunction}
sgn_\epsilon(x_i) =
\begin{cases}
    -1 & \text{if } x_i < -\epsilon \\
    x_i/\epsilon & \text{if } -\epsilon \leq x_i \leq \epsilon \\
    1 & \text{if } x_i >\epsilon.
\end{cases}\] \\
In a two agent \opiniondynamic{} with some network influence $b > 0$, and with the agent's influence on the other agent given by $a_{12} = a_{21} = a$ for some $a > 0$. There are five steady states:
They are $(1,1),$ $(-1,-1), (0,0), (\frac{b}{2 a+b}, \frac{-b}{2 a+b})$ and $(\frac{-b}{2 a+b}, \frac{b}{2 a+b})$, as shown in~\cite{candogan2022social}.
We generalize this prove to $n$ agents in Theorem~\ref{thm:signInClique}, but it is useful to first understand the case of two agents. 

We visualize the aforementioned equilibria in Figure \ref{fig:plotSignFun}, the green points show the consensus points, and the red points show the PD points. 
The proof sketch for two agents is as follows:

For $\epsilon$ small enough, we show that $F_1$ and $G$ intersect in the point $(\frac{b}{2 a+b}, \frac{-b}{2 a+b})$. Let \( x_i = \frac{b}{2 a+b} \) and \( x_j = \frac{-b}{2 a+b} \). The expression for \( G(x_i(t),x_j(t),a,b) \) is
\[ G(x_i(t),x_j(t),a,b) = \frac{a (x_i(t) - x_j(t))}{b} + x_i(t) \]
Substituting the values of \( x_i(t) \) and \( x_j(t) \), we obtain
\[G(x_i(t),x_j(t),a,b) = \frac{a (\frac{b}{2 a+b} - \left(\frac{-b}{2 a+b}\right))}{b} + \frac{b}{2 a+b}=1 \]
Thus, \( G(x_i(t),x_j(t),a,b) = f(\frac{b}{2 a+b} )\),if $\epsilon$ is sufficiently small.
For analogue reasons it holds that \( H(x_i(t),x_j(t),a, b) = f(\frac{-b}{2 a+b} )\) if $\epsilon$ is small enough.
\begin{figure}[ht]
    \centering
    \begin{minipage}[b]{0.32\textwidth} 
        \includegraphics[width=\textwidth]{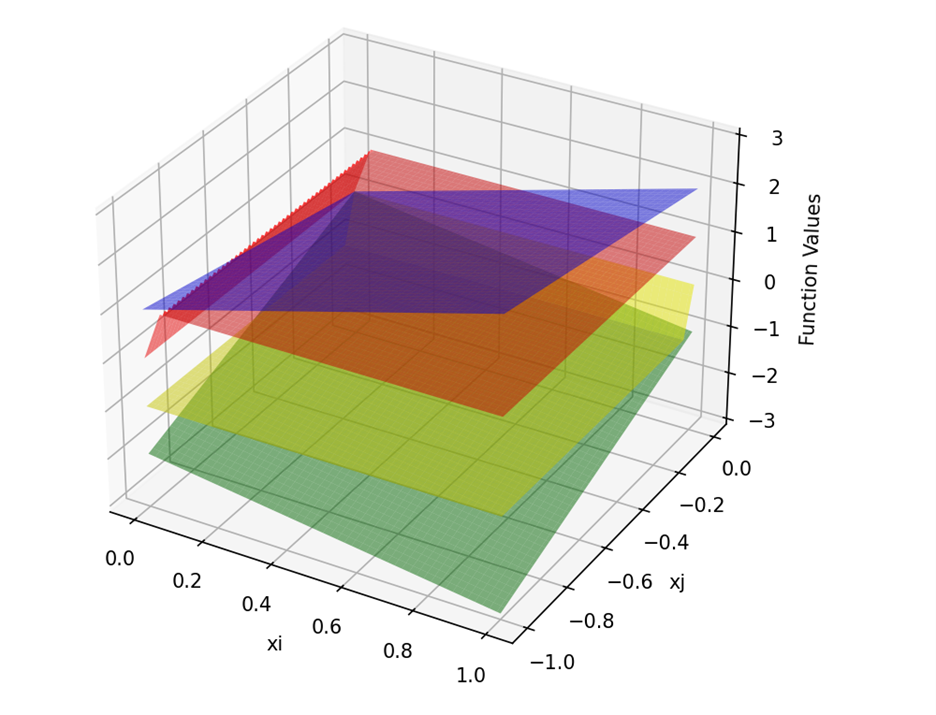}
        \caption{A depiction of the sign function and polarization planes: $G$ (purple), $F_1$ (red), $F_2$ (yellow) and $H$ (green).}
        \label{fig:SignFuncIntersection-1}
    \end{minipage}
    \hfill 
    \begin{minipage}[b]{0.32\textwidth}
        \includegraphics[width=\textwidth]{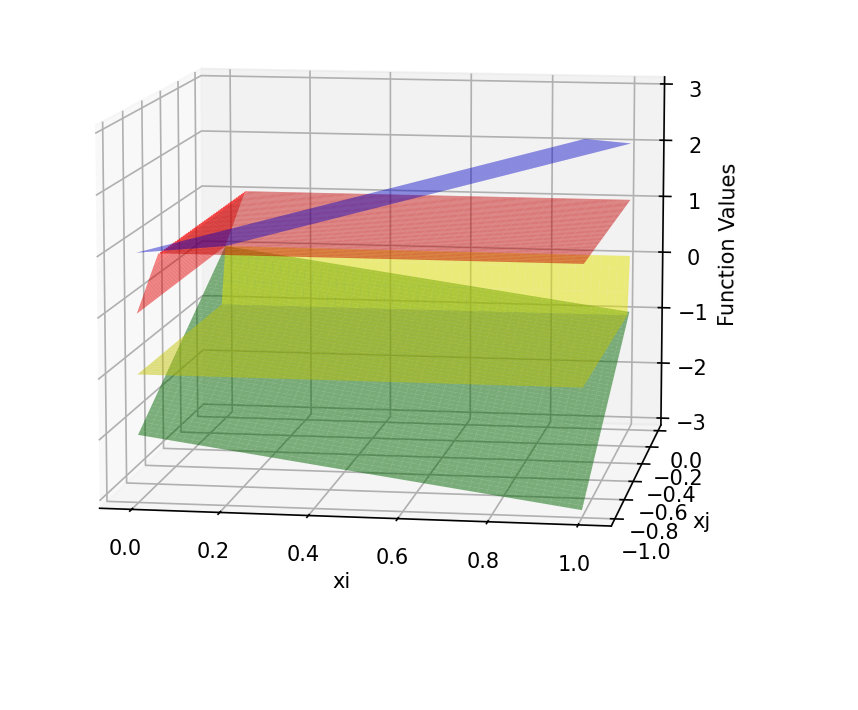}
        \caption{A depiction of the sign function and polarization planes: $G$ (purple), $F_1$ (red), $F_2$ (yellow) and $H$ (green) from a different viewing angle}
        \label{fig:SignFuncIntersection-2}
    \end{minipage}
    \hfill 
    \begin{minipage}[b]{0.32\textwidth}
        \includegraphics[width=\textwidth]{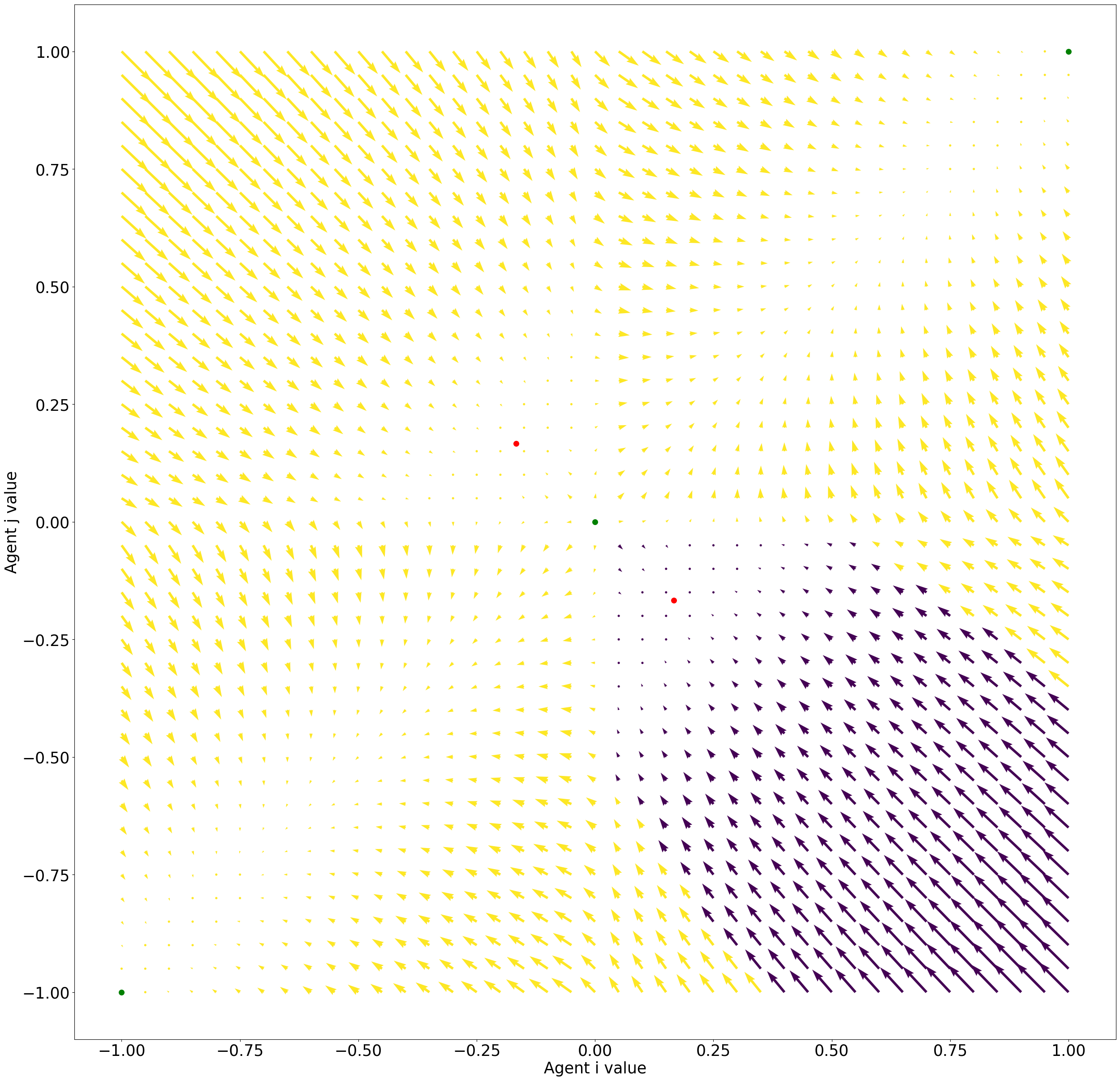}
        \caption{Vector field plot of Sign Function in two agents system where b = 0.4 and a=1. Green points indicate consensus while red ones represent persistent disagreement}
        \label{fig:plotSignFun}
    \end{minipage}
\end{figure}
\begin{theorem}
\label{thm:signInClique}
Consider an \opiniondynamic{} $x$ with $n$ agents and with 
\( f(x_i) = sgn_\epsilon(x_i)\) for $\lim \epsilon \rightarrow 0$. 
For $v\in [-1,1]$, let $s(v)=|\{ i~|~x_i=v\}|$ be the number of agents with value $v$.
The set of states that are equilibria for every $\varepsilon$ is the union of the following five sets. 
\begin{enumerate}
\item $\{ (0,0, \ldots , 0)\}$, $\{ (1,1, \ldots , 1)\}$,$\{ (-1,-1, \ldots , -1)\}$
\item For every $k \in [n-1]$:\\ 
$\left\{ (x_i)_{i\in [n]} ~\middle|~s\left(\frac{2a k^+ - a n + b n}{a n + b n}\right) = k \text{ and } s\left(\frac{-(a n - 2a k^+ + b n)}{a n + b n}\right) = n-k \right\}$
\item For every  $k \in [n-1]$ and with $2k< n$:\\ 
$\left\{ (x_i)_{i\in [n]} ~\middle|~s\left(\frac{b}{b+a} \right) = k, s\left(\frac{-b}{b+a}\right) = k \text{ and } s\left(0\right) = n-2k  \right\}$
\end{enumerate}
\end{theorem}  

\begin{proof}
We will show that all points named in the Theorem are a steady state for $\epsilon$ small enough. Moreover, we will show that these are the only steady-state points.

In any system, we will have three consensus steady points which are  (1, 1, 1, \ldots, 1), (-1, -1, -1, \ldots, -1), and (0, 0, 0, \ldots, 0) $\in\mathbb{R}^n$. From Theorem \ref{thm:AN}, we obtain that these states are steady states since $f(1)=1$.

Next we show that there can be no other steady state where all entries are positive beside $(1, \ldots, 1) \in\mathbb{R}^n$.

\begin{itemize}
\item Case 1:
All agents are positive. Let agent $j \in [n]$ be an arbitrary positive agent. It's derivative is given by the following. 
\begin{align*}
\dot{x}_i(t) = \frac{a}{n-1} \sum_{j \in N(i)} (x_j(t) - x_i(t)) + b(1 - x_i(t)) \\
=\frac{a}{n-1} \sum_{j \in N(i)} x_j(t) - \frac{a}{n-1} \sum_{j \in N(i)} x_i(t) + b(1 - x_i(t))\\
=\frac{a}{n-1} \sum_{j \in N(i)} x_j(t) - a x_i(t) + b - b x_i(t)\\
= \frac{a}{n-1} \sum_{j \in N(i)} x_j(t) + b - (a + b) x_i(t)\\
\end{align*}
A steady state would fulfill $\dot{x_i}(t)=0$.
Thus in steady state we have
\[
x_i \cdot (b + a) = \frac{a}{n-1} \sum_{j \in N(i)}x_j + b \quad \forall i \in [n].
\]
Thus, since the r.h.s. in all the equations above yields the same value, all the $x_i$ must have the same value, since for all $i$ and $j$ we have that $x_i(b+a)=x_j(b+a)$.
We now calculate said value $v$. We have $v (b+a) = a v + b$. Thus, $vb = b$ which implies $v=1$ using that $b>0$. 

Therefore, there is only one stable state where all agents have value 1. Otherwise, the derivative is not zero which means we are not in a steady state. 

\item Case 2: With a reasoning that is analogue to case 1 it follows that the only steady state where all agents are negative is the steady state $(-1, \ldots, -1)$.

\item Case 3: Here we aim to describe all steady states that have positive and negative entries (and might have additional entries that are 0). We will assume that $\epsilon>0$ is arbitrary small. 
First, note that we have a set of $n$ equations that are necessary for an equilibrium, where for agent $i$ we have
\begin{equation*}
\frac{a}{n-1}\cdot \sum_{j \in [n]\setminus \{ i\} } (x_j - x_i) + b(f(x_i)-x_i) = 0
\end{equation*}
Rewriting this gives
\begin{equation}\label{eq:eqi}
\frac{a}{n-1} \cdot \sum_{j \in [n] } x_j - (n-1) x_i  + b(f(x_i)-x_i) = 0
\end{equation}
Consider indices $i'$ and $j'$ such that $x_{i'} > 0$ and $x_{j'}<0$.
From \eqref{eq:eqi} it must hold that
\begin{equation}\label{eq:askd}  - a x_{i'}  + b(1-x_{i'}) = - a x_{j'}  + b(-1-x_{j'}). 
\end{equation}
This implies that for every positive value, there can only be one negative value that fulfills \eqref{eq:askd} since the degree of freedom is one.
Likewise, for every negative value there can only be one positive value.
Thus, there can only be one negative and one positive value.
Let $k^+$ be the number of agents with the positive value and $k^-$ the number of agents with the negative value. 
We now argue that these are the values $y^+$ and $y^-$.

For an agent with a  positive value we get
\begin{equation}\label{eq:as1} 
\frac{ak^-}{n-1}\left(y^- - y^+\right) + \frac{a(n - k^- - k^+ -1)}{n-1}\left(0 - y^+\right) + b(1 - y^+)=0
\end{equation}
and for an agent with a negative value we get
\begin{equation}\label{eq:as2}
\frac{ak^+}{n-1}\left(y^+ - y^-\right) + \frac{a(n - k^- - k^+ -1)}{n-1}\left(0 - y^-\right) + b(-1 - y^-) = 0.
\end{equation}
Finally, for an agent with value $0$, we have

\begin{equation}\label{eq:as3} \frac{ak^+}{n-1}(y^+) + \frac{ak^-}{n-1}(y^-) = 0.
\end{equation}

We now distinguish between three cases:
\begin{enumerate}
\item Consider the case $k^+ + k^- = n$. 
Using the values $y^+= \frac{2a k^+ - a n + b n}{a n + b n}$ and $y^-= \frac{-(a n - 2a k^+ + b n)}{a n + b n}$ we can verify that
\eqref{eq:as1} and  \eqref{eq:as2} hold.
Note that \eqref{eq:as3} is not require to hold since all non-zero, also in this case we have to replace $a(n - k^- - k^+)$ by $a(n - k^- - k^+-1)$ to exclude the agent herself.
This characterizes the states of case 2.
\item Consider the case $k^+ + k^- < n$ and $k^+=k^-$. This means that there are $n-k^+-k^-$ agents with a value of 0.Using the value $y^+=\frac{b}{b+a}$ for agents with positive value and and the value $y^-=-\frac{b}{b+a}$ for the agents with negative value, we see that  \eqref{eq:as1},  \eqref{eq:as2} and  \eqref{eq:as3} hold assuming. This characterizes the states of case 3.
\item  Consider the case $k^+ + k^- < n$ and $k^+ \neq k^-$. Plugging \eqref{eq:as3} into the other equations yields that the positive value must be $b/(b+a)$ and the negative value must be  $-b/(b+a)$. However, plugging these values into  \eqref{eq:as3}  and using that $b>0$ implies that $k^+=k^-$ which is a contradiction. Hence, this case cannot occur.
\end{enumerate}
The above characterizes all possible cases for PD.
\end{itemize}
\end{proof}
\subsection{Examples of Harmonizing Platform Influence Functions(HPIF)}
\label{sec:HPIF}
This section presents two examples of harmonizing platform influence functions. 
\subsubsection{The Linear Harmonizing Function}
Fix \(0 < \alpha < 1\), the first platform function is
\(    \label{eq:LinearHarmonizingFunction}
    f(x_i) = \alpha x_i
\)    
 Since $0 < \alpha < 1$, one can show (Theorem~\ref{thm:LinearHarmonizing}) that the system will always converge to the equilibrium $(0,0,\dots,0)$.
From the point of view of a platform such a function still keeps users on the platform, because the content they see, governed by $f(\cdot)$, is close to their own. However, it gradually pushes users towards a state of consensus.
It is a harmonizing function, therefore, it does not intersect with the polarization planes as shown in Figures \ref{fig:LinearHarmonizingFunction} and \ref{fig:LinearHarmonizingFunction2}.
To show that this function never intersects with the plane $G$. Lets consider \( x_i \in (0, 1] \) and \( x_j \in [-1,0) \) to intersects with $G$, the following equation has to hold true.
\(\alpha x_i = \frac{a(x_i - x_j)}{b} + x_i\)
Subtract \(x_i\) from both sides:
\((\alpha -1)x_i = \frac{a(x_i - x_j)}{b}\)
This can never be true as the left-hand side of the equation is always negative and right-hand side of the equation is always positive. 

\begin{theorem}
\label{thm:LinearHarmonizing}
Consider a clique of n agents and the following platform signal function. There is only a stable solution which is the state where all agents have value $0$, i.e., for all $i\in [n]$, $x_i=0$.
\( f(x_i) = \alpha x_i \)
\end{theorem}

\begin{proof}
Consider a complete graph \( K_n \) with \( n \) agents. The platform function is given by \( f(x_i) = \alpha x_i \), where \( 0 < \alpha < 1 \). To show that there is only one steady state we must show that this function does not intersect with the polarization plane \( G \), except when all agents converge to the only steady state, which is \( x_i = 0 \) for all \( i \in [n] \).

The evolution of the agent \( i \)'s with maximum opinion is governed by the following equation:
\[
\dot{x}_{i}(t) = \frac{a}{n-1} \sum_{j \in N(i)} \left( x_j(t) - x_i(t) \right) + b \left( f(x_i(t)) - x_i(t) \right),
\]
where \( N(i) \) denotes the set of all agents except \( i \). At equilibrium, where \( \dot{x}_i(t) = 0 \), this equation simplifies to:
\[
f(x_i(t)) = \frac{(a + b) x_i(t) - a \bar{x}_i(t)}{b},
\]
where \( \bar{x}_i = \frac{1}{n-1} \sum_{j \in N(i)} x_j \) is the average opinion of the neighboring agents of \( i \). 
To prove that \( f(x_i(t)) = \alpha x_i(t) \) does not intersect with the polarization plane \( G \) (except at \( x_i = 0 \)), we first express the polarization condition in terms of the function \( G(x_i(t), \bar{x}_i(t), a, b) \):
\[
G(x_i, \bar{x}_i, a, b) = \frac{(a + b) x_i - a \bar{x}_i}{b}.
\]
For equilibrium, this equation must satisfy:
\[
f(x_i) = G(x_i, \bar{x}_i, a, b).
\]
Substituting \( f(x_i) = \alpha x_i \), we get:
\[
\alpha x_i = \frac{(a + b) x_i - a \bar{x}_i}{b}.
\]
\[
(\alpha - 1) x_i = \frac{a (x_i - \bar{x}_i)}{b}.
\]
Since \( x_i > 0 \) and has the maximum value. The left-hand side \( (\alpha - 1) x_i \) is negative since \( \alpha < 1 \) and \( x_i > 0 \), while the right-hand side \( \frac{a (x_i - \bar{x}_i)}{b} \) is positive because \( x_i > \bar{x}_i \) (for neighbors with lower opinions). Thus, this equation cannot hold unless \( x_i = 0 \), as the left-hand side and right-hand side have opposite signs.
Hence, the platform function \( f(x_i) = \alpha x_i \) does not intersect with the polarization plane \( G \), except when all agents converge to the consensus state \( x_i = 0 \). This shows that the only stable equilibrium in the complete graph \( K_n \) is when all agents' opinions converge to \( 0 \).
\end{proof}
\vspace{-15pt} 
\begin{figure}[ht]
    \centering
    \begin{minipage}[b]{0.48\textwidth} 
        \includegraphics[width=\textwidth]{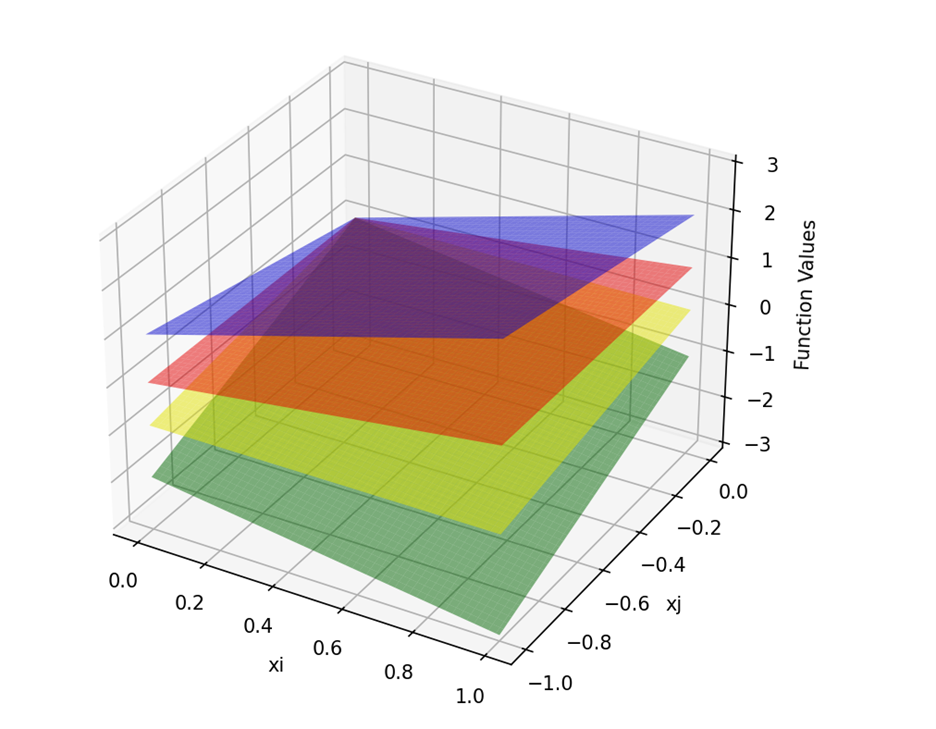}
        \caption{The Linear Harmonizing Function and Polarization Planes with $\alpha = 0.9$}
        \label{fig:LinearHarmonizingFunction}
    \end{minipage}
    \hfill 
    \begin{minipage}[b]{0.48\textwidth}
        \includegraphics[width=\textwidth]{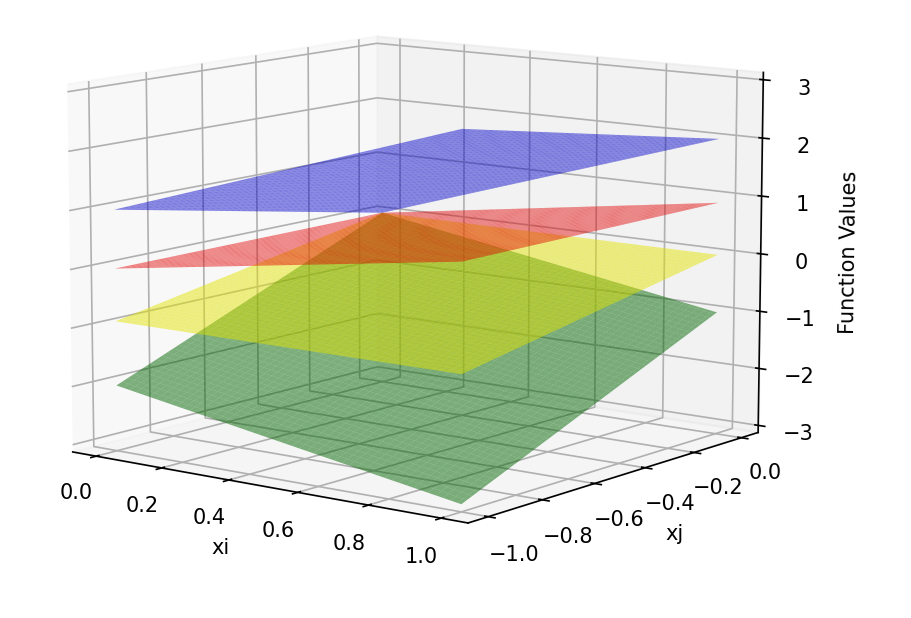}
        \caption{The Linear Harmonizing Function and Polarization Planes $\alpha = 0.9$ and changing the viewing angle}
        \label{fig:LinearHarmonizingFunction2}
    \end{minipage}
\end{figure}
\subsubsection{Anti Sign Function}
\label{AntiSignFunction}
This platform function presents the user with content that opposes their initial opinion. This can be beneficial in promoting diverse perspectives and reducing polarization. For example, if a user expresses agreement with a particular topic or viewpoint, the platform will show them content that presents an opposing viewpoint or opinion. This function can encourage users to consider alternative perspectives and make more informed decisions. In a two-agent system and some network influence $b>0$ and peer influence $a$, then according to Equation \ref{eq:difeq}, there is one stable state as shown in Figure \ref{fig:AntiSignFunction}. 
\begin{theorem}
\label{thm:antiSignFunction}
Consider a clique of $n$ agents and the $antiSgn_\epsilon(x_i)$ as the platform function. There is only a stable solution which is $(0,0,\dots, 0)$.
\end{theorem}
\begin{proof} 
The proof is similar to Theorem \ref{thm:LinearHarmonizing} where there is only one steady state. 
\end{proof}
\begin{figure}[ht]
\centering
\begin{minipage}{0.5\textwidth}
\centering
\includegraphics[width=\linewidth]{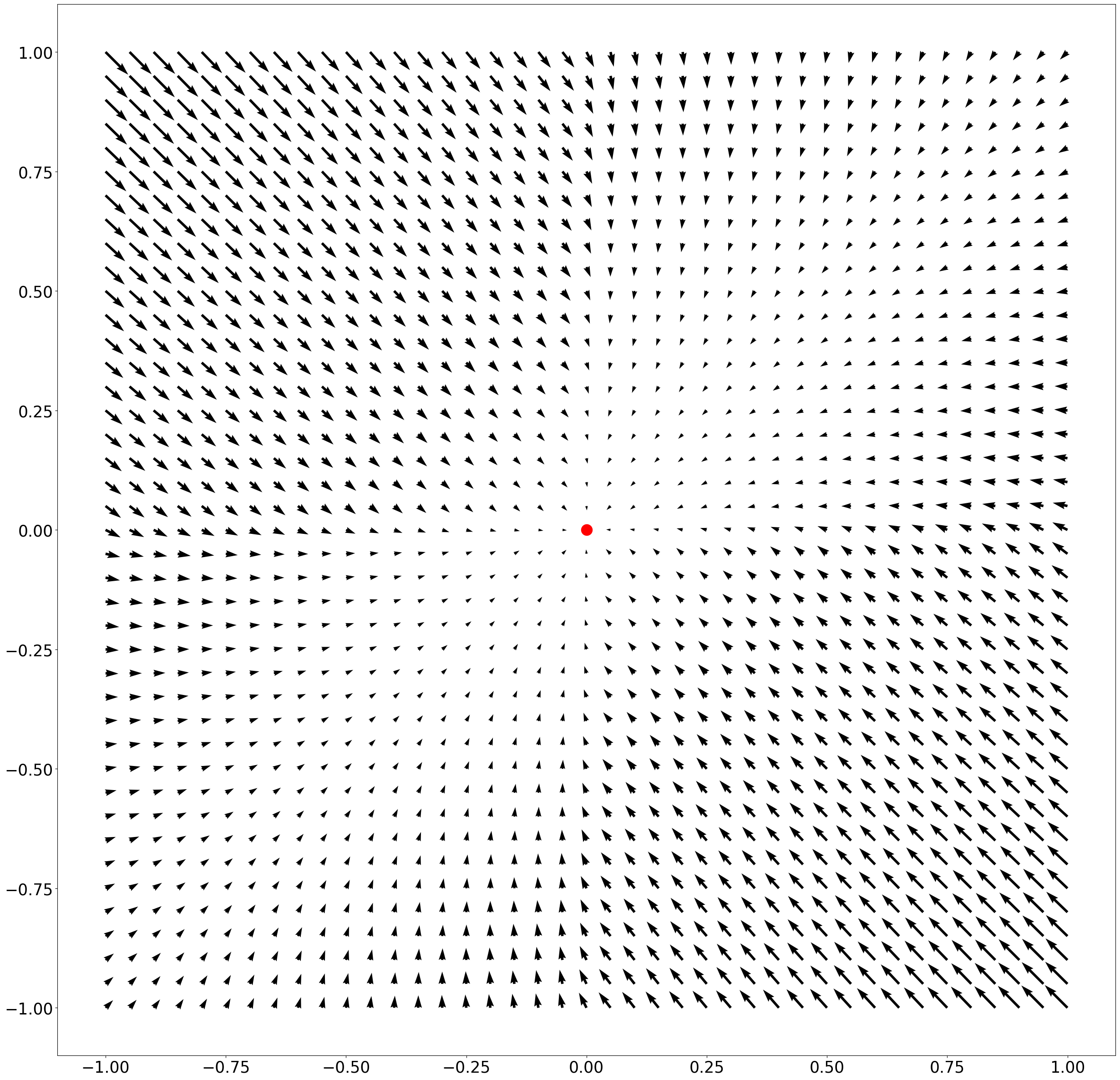}
\caption{Vector field plot of Anti Sign Function in two agents system showing all points are converging to the origin.}
\label{fig:AntiSignFunction}
\end{minipage}%
\hfill
\begin{minipage}{0.45\textwidth}
The anti-sign function is defined as the following. 
\[\label{eq:antiSignFunction}
antiSgn_\epsilon(x_i) =
\begin{cases}
    1 & \text{if } x_i < -\epsilon \\
    x_i/\epsilon & \text{if } -\epsilon \leq x \leq \epsilon \\
    -1 & \text{if } x >\epsilon
\end{cases}\]\\
Please note that the anti sign function violates the sign preservation property. Moreover, realistically speaking a platform would avoid showing a user constantly content that is so far from its opinion as it will likely cause the user to feel alienized and disengage with the platform.
\end{minipage}
\end{figure}
\clearpage
\section{Conclusion and Future Work}
\label{sec:Conclusion}
In this paper, we show how individual viewpoints shaped by social interactions are influenced by the platform by studying the platform influence function. We study when strong consensus and persistent disagreement is possible by giving the necessary and sufficient conditions for both. To achieve this we study a concept we coin polarisation planes that reveal interesting geometric properties. Moreover, we generalise these results to clique of n agents, also we proved that our results can be generalized to stochastic block model of two blocks. 
Studying the regions of attraction - meaning to which equilibrium state the system converges depending on the initial state - remains an interesting open question.
We also provide some functions that avoid persistent disagreement and it would be interesting to see if such functions help in the real world to combat polarization in social networks

\paragraph{Acknowledgements}
We are very grateful to Laura Vargas Koch and Víctor Verdugo for their help in writing the paper.

\clearpage
\bibliographystyle{plain} 
\bibliography{refs}    

\section{Appendix} 
\appendix
\section{Proof of Theorem in SBM}
\label{sec:appendix-SBM}
Here is the proof of the Theorem~\ref{thm:SBM} mentioned in the Section~\ref{sec:SBM}.
To prove Theorem ~\ref{thm:SBM}, we follow the structure of the proof used in the proof of Theorem 2 in~\cite{candogan2022social} and rely on similar Lemmas and notations, with slight modifications about using general platform function that satisfies the properties in~\ref{sec:properties} instead of the sign function.
\subsection*{Proof of Theorem \ref{thm:SBM}}
To establish Theorem \ref{thm:SBM}, we define \( \mathcal{L}_{n} \) as the set of agents in the left block and \( \mathcal{R}_{n} \) as the set of agents in the right block. For each agent \( i \), \( L_{n, i} \) represents the number of its left connections, while \( R_{n, i} \) represents the number of its right connections.

The proof of Theorem \ref{thm:SBM} consists of two lemmas.
The first lemma holds from lemma 4 in~\cite{candogan2022social}, the lemma states that for a large number of agents in the stochastic block model, the number of connections each agent has within their own block and with the other block is tightly concentrated around a predefined expected values. Specifically, as the network size $n$ increases, the probability that any agent’s number of connections diverges from these expected values is almost zero. This ensures that most agents has the expected range of connections, which is essential for proving the overall stability and predictability of the opinion dynamics in the network.
This lemma is identical to Lemma 4 in~\cite{candogan2022social}
\begin{lemma}
\label{connection-lemma}
Consider the sets
$$
S_{\delta, n} = [(1-\delta) p(n) n, (1+\delta) p(n) n] \text{ and } D_{\delta, n} = [(1-\delta) q(n) n, (1+\delta) q(n) n]
$$
Then
$$
\lim_{n \rightarrow \infty} \mathbb{P}(R_{n, i} \notin D_{\delta, n}, \forall i \in \mathcal{L}_{n}) = \lim_{n \rightarrow \infty} \mathbb{P}(L_{n, i} \notin S_{\delta, n}, \forall i \in \mathcal{L}_{n}) = 0
$$
and
$$
\lim_{n \rightarrow \infty} \mathbb{P}(R_{n, i} \notin S_{\delta, n}, \forall i \in \mathcal{R}_{n}) = \lim_{n \rightarrow \infty} \mathbb{P}(L_{n, i} \notin D_{\delta, n}, \forall i \in \mathcal{R}_{n}) = 0
$$
for all \( \delta \in (0, 1) \).
\end{lemma}

\begin{proof} 
Let \( \delta \in (0, 1), n \geq 1 \) and consider a left agent \( i \in \mathcal{L}_{n} \). Note that the number of left connections \( L_{n, i} \) that agent \( i \) has is a binomial random variable with a probability of success \( p(n) \) and \( n \) trials. Applying Theorem 12.6 in~\cite{blum2020foundations} we have 
\begin{equation}
\mathbb{P}(L_{n, i} \notin [(1-\delta) p(n) n, (1+\delta) p(n) n]) \leq 3 \exp(-\delta^{2} p(n) n / 8) 
\end{equation}
Hence, using the union bound, we have
$$
\mathbb{P}(L_{n, i} \notin [(1-\delta) p(n) n, (1+\delta) p(n) n], \forall i \in \mathcal{L}_{n}) \leq 3 n \exp(-\delta^{2} p(n) n / 8)
$$
The fact that \( p(n) \) is \( \omega(\ln (n) / n) \) implies that \( 3 n \exp(-\delta^{2} p(n) n / 8) \) converges to 0 as \( n \rightarrow \infty \). To see this note that
$$
n \exp(-\delta^{2} p(n) n / 8) = \exp(\ln (n) (1 - \delta^{2} p(n) n / 8 \ln (n))) = n^{1 - \delta^{2} p(n) n / 8 \ln (n)}
$$
converges to 0 when \( n \) and \( p(n) n / \ln (n) \) converge to \( \infty \).
Similarly, we can deduce that
$$
\mathbb{P}(R_{n, i} \notin [(1-\delta) q(n) n, (1+\delta) q(n) n], \forall i \in \mathcal{L}_{n}) \leq 3 n \exp(-\delta^{2} q(n) n / 8)
$$
converges to 0 as \( n \rightarrow \infty \). The proof of the lemma for the right agents follows from the same arguments as the above arguments for the left agents.\end{proof}
Define the event \(C_{\delta, n}\) which is the set of all adjacency matrices such that every agent has between \( (1-\delta)p(n)n \) and \( (1+\delta)p(n)n \) connections within the same block, and between \( (1-\delta)q(n)n \) and \( (1+\delta)q(n)n \) connections with the different block.
\begin{equation}
\begin{aligned}
C_{\delta, n} = \{A_{n} \in \boldsymbol{A}_{n}: &\ R_{n, i} \in D_{\delta, n}, L_{n, i} \in S_{\delta, n}, \forall i \in \mathcal{L}_{n}, \\
& \text{and} \ R_{n, j} \in S_{\delta, n}, L_{n, j} \in D_{\delta, n}, \forall j \in \mathcal{R}_{n} \}
\end{aligned}
\end{equation}
The proof of Lemma \ref{connection-lemma}, along with the union bound, shows that for any \( \delta \in (0, 1) \), the probability that a randomly chosen adjacency matrix belongs to the set \( C_{\delta, n} \) approaches 1 as \( n \) grows. In other words, as \( n \) increases, it becomes almost certain that the adjacency matrices will satisfy the conditions defining \( C_{\delta, n} \).
For \( 0 < \delta < 1 \) and an integer \( n \), consider the following differential equations where \(f(.)\) is any platform function that follows the platform properties in~\ref{sec:properties}:
\begin{equation}
\begin{aligned}
& \dot{\bar{x}}_{L}(t) = a(\bar{x}_{R}(t) - \bar{x}_{L}(t)) \frac{(1+\delta) q(n)}{(1-\delta)(p(n) + q(n))} + b(f(\bar{x}_{L}(t)) - \bar{x}_{L}(t)) \\
& \underline{\dot{x}}_{L}(t) = a(\underline{x}_{R}(t) - \underline{x}_{L}(t)) \frac{(1-\delta) q(n)}{(1+\delta)(p(n) + q(n))} + b(f(\underline{x}_{L}(t)) - \underline{x}_{L}(t)) \\
& \dot{\bar{x}}_{R}(t) = a(\bar{x}_{L}(t) - \bar{x}_{R}(t)) \frac{(1-\delta) q(n)}{(1+\delta)(p(n) + q(n))} + b(f(\bar{x}_{R}(t)) - \bar{x}_{R}(t)) \\
& \underline{\dot{x}}_{R}(t) = a(\underline{x}_{L}(t) - \underline{x}_{R}(t)) \frac{(1+\delta) q(n)}{(1-\delta)(p(n) + q(n))} + b(f(\underline{x}_{R}(t)) - \underline{x}_{R}(t))
\end{aligned}
\label{eq:upper-lower-ODEs}
\end{equation}
with the initial conditions \( \bar{x}_{i}(0) = \max X_{i} \) and \( \underline{x}_{i}(0) = \min X_{i} \) for \( i = L, R \). In other words, \( \bar{x}_{R}(0) \) and \( \bar{x}_{L}(0) \) correspond to the maximum initial opinions for the right and left agents, respectively, while \( \underline{x}_{R}(0) \) and \( \underline{x}_{L}(0) \) represent the minimum initial opinions for the right and left agents. Since the existence of a solution to \( \boldsymbol{x}_{n, A_{n}, \boldsymbol{x}_{0}}(t) \) for any chosen adjacency matrix \( A_{n} \) and initial opinion vector \( \boldsymbol{x}_{0} \) is guaranteed from standard differential equation theory~\cite{cortes2008discontinuous}.
Let
$$
a_{1} = \frac{a(1-\delta) q(n)}{(1+\delta)(p(n) + q(n))} \text{ and } a_{2} = \frac{a(1+\delta) q(n)}{(1-\delta)(p(n) + q(n))}
$$
In the following Lemma \ref{lem:boundedODE}, it shows that the solutions of the ODEs given in \ref{eq:upper-lower-ODEs} bound the agents' opinions \( \boldsymbol{x}_{n, A_{n}, \boldsymbol{x}_{0}}(t) \) for adjacency matrices \( A_{n} \) that belong to \( C_{\delta, n} \) and initial opinions that belong to \( (X_{L} \times X_{R})^{n} \).
This lemma is inspired from Lemma 5 in~\cite{candogan2022social}, we just instead of using the sign function, we use \(f(.)\) which represent any platform function and it is shown to hold under the given proprieties in ~\ref{sec:properties}.
\begin{lemma}\label{lem:boundedODE}
Suppose that \( \bar{x}_{L}(t) < \underline{x}_{R}(t) \) for \( t \geq 0 \).\\
(i) We have \( \underline{x}_{R}(t) \leq \bar{x}_{R}(t) \) and \( \underline{x}_{L}(t) \leq \bar{x}_{L}(t) \) for all \( t \geq 0 \).\\
(ii) For all \( A_{n} \in C_{\delta, n}, \boldsymbol{x}_{0} \in (X_{L} \times X_{R})^{n} \) and all \( t \geq 0 \) we have
\begin{equation}
\underline{x}_{L}(t) \leq x_{i, n, A_{n}, \boldsymbol{x}_{0}}(t) \leq \bar{x}_{L}(t), \forall i \in \mathcal{L}_{n}, \text{ and } \underline{x}_{R}(t) \leq x_{j, n, A_{n}, \boldsymbol{x}_{0}}(t) \leq \bar{x}_{R}(t), \forall j \in \mathcal{R}_{n} 
\label{eq:main-inequality-equation}
\end{equation}
\end{lemma}
\begin{proof}(i) Assume in contradiction that there exists a \( t \geq 0 \) such that \( \underline{x}_{R}(t) > \bar{x}_{R}(t) \) or \( \underline{x}_{L}(t) > \bar{x}_{L}(t) \). Note that \( t > 0 \). Let \( t_{1} = \inf \{ t \in [0, \infty) : \underline{x}_{R}(t) > \bar{x}_{R}(t) \text{ or } \underline{x}_{L}(t) > \bar{x}_{L}(t) \} \). By the contradiction assumption and the continuity of the solutions of the ordinary differential equations given in Equation ~\ref{eq:upper-lower-ODEs}, \( t_{1} \) is finite, and we have \( \underline{x}_{R}(t_{1}) = \bar{x}_{R}(t_{1}) \) and \( \underline{x}_{L}(t_{1}) \leq \bar{x}_{L}(t_{1}) \) or \( \underline{x}_{L}(t_{1}) = \bar{x}_{L}(t_{1}) \) and \( \underline{x}_{R}(t_{1}) \leq \bar{x}_{R}(t_{1}) \). Assume without loss of generality that \( \underline{x}_{R}(t_{1}) = \bar{x}_{R}(t_{1}) \) and \( \underline{x}_{L}(t_{1}) \leq \bar{x}_{L}(t_{1}) \). We have

$$
\begin{aligned}
\dot{\bar{x}}_{R}(t_{1}) & = a(\bar{x}_{L}(t_{1}) - \bar{x}_{R}(t_{1})) \frac{(1-\delta) q(n)}{(1+\delta)(p(n) + q(n))} + b(f(\bar{x}_{R}(t_{1})) - \bar{x}_{R}(t_{1})) \\
& > a(\underline{x}_{L}(t_{1}) - \underline{x}_{R}(t_{1})) \frac{(1+\delta) q(n)}{(1-\delta)(p(n) + q(n))} + b(f(\underline{x}_{R}(t_{1})) - \underline{x}_{R}(t_{1})) = \underline{\dot{x}}_{R}(t_{1})
\end{aligned}
$$

so that \( \underline{x}_{R}(t) < \bar{x}_{R}(t) \) for \( t > t_{1} \) that is close enough to \( t_{1} \).

From continuity, if \( \underline{x}_{L}(t_{1}) < \bar{x}_{L}(t_{1}) \) then \( \underline{x}_{L}(t) < \bar{x}_{L}(t) \) for \( t > t_{1} \) that is close enough to \( t_{1} \). If \( \underline{x}_{L}(t_{1}) = \bar{x}_{L}(t_{1}) \) then

$$
\begin{aligned}
\dot{\bar{x}}_{L}(t_{1}) & = a(\bar{x}_{R}(t_{1}) - \bar{x}_{L}(t_{1})) \frac{(1+\delta) q(n)}{(1-\delta)(p(n) + q(n))} + b(f(\bar{x}_{L}(t_{1})) - \bar{x}_{L}(t_{1})) \\
& > a(\underline{x}_{R}(t_{1}) - \underline{x}_{L}(t_{1})) \frac{(1-\delta) q(n)}{(1+\delta)(p(n) + q(n))} + b(f(\underline{x}_{L}(t_{1})) - \underline{x}_{L}(t_{1})) = \underline{\dot{x}}_{L}(t_{1})
\end{aligned}
$$

so that \( \underline{x}_{L}(t) < \bar{x}_{L}(t) \) for \( t > t_{1} \) that is close enough to \( t_{1} \). Hence, \( \underline{x}_{R}(t) < \bar{x}_{R}(t) \) and \( \underline{x}_{L}(t) < \bar{x}_{L}(t) \) for \( t > t_{1} \) that is close enough to \( t_{1} \), which is a contradiction to the definition of \( t_{1} \). Note that this comparison is possible because \( f(.)\)is monotonically increasing.
\begin{align*}
\dot{x}_{i, n, A_{n}, \boldsymbol{x}_{0}}(t) & = \frac{a}{|N(i)|} \sum_{j \in N(i)} (x_{j, n, A_{n}, \boldsymbol{x}_{0}}(t) - x_{i, n, A_{n}, \boldsymbol{x}_{0}}(t)) \\
& \quad + b(f(x_{i, n, A_{n}, \boldsymbol{x}_{0}}(t)) - x_{i, n, A_{n}, \boldsymbol{x}_{0}}(t)) \\
& \leq \frac{a}{|N(i)|} \sum_{j \in N(i) \cap \mathcal{L}_{n}} (x_{j, n, A_{n}, \boldsymbol{x}_{0}}(t) - x_{i, n, A_{n}, \boldsymbol{x}_{0}}(t)) \\
& \quad + b(f(x_{i, n, A_{n}, \boldsymbol{x}_{0}}(t)) - x_{i, n, A_{n}, \boldsymbol{x}_{0}}(t)) \\
& \leq \frac{a |N(i) \cap \mathcal{L}_{n}|}{|N(i)|} (\bar{x}_{L}(t) - x_{i, n, A_{n}, \boldsymbol{x}_{0}}(t)) \\
& \quad + b(f(x_{i, n, A_{n}, \boldsymbol{x}_{0}}(t)) - x_{i, n, A_{n}, \boldsymbol{x}_{0}}(t)) \\
& \leq a \frac{(1-\delta) q(n)}{(1+\delta)(p(n) + q(n))} (\bar{x}_{L}(t) - x_{i, n, A_{n}, \boldsymbol{x}_{0}}(t)) \\
& \quad + b(f(x_{i, n, A_{n}, \boldsymbol{x}_{0}}(t)) - x_{i, n, A_{n}, \boldsymbol{x}_{0}}(t))
\end{align*}

The first inequality holds because \( x_{i^{\prime}, n, A_{n}, \boldsymbol{x}_{0}}(t) - x_{i, n, A_{n}, \boldsymbol{x}_{0}}(t) \leq 0 \) for every right agent \( i^{\prime} \in \mathcal{R}_{n} \). The second inequality holds because \( x_{j, n, A_{n}, \boldsymbol{x}_{0}}(t) \leq \bar{x}_{L}(t) \) for every left agent \( j \in \mathcal{L}_{n} \). The third inequality holds because \( A_{n} \in C_{\delta, n} \) implies that the cardinality of \( N(i) \cap \mathcal{L}_{n} \) is at least \( (1-\delta) q(n) n \) and the cardinality of \( N(i) \) is at most \( (1+\delta)(p(n) + q(n)) n \) and because \( \bar{x}_{L}(t) < x_{i, n, A_{n}}(t) \).

We conclude that \( \dot{x}_{i, n, A_{n}}(t) \leq \bar{U}(t, x_{i, n, A_{n}}(t)) \) where

$$
\bar{U}(t, y) = a (\bar{x}_{L}(t) - y) \frac{(1-\delta) q(n)}{(1+\delta)(p(n) + q(n))} + b(f(y) - y)
$$

By applying Theorem 4.1 from \cite{hartman2002ordinary}, we can conclude that \( x_{i, n, A_{n}}(t) \leq \bar{x}_{R}(t) \) holds over the interval \( [t, t + \delta'] \), for some \( \delta' > 0 \) independent of \( t \). A similar reasoning shows that the remaining inequalities in the lemma (refer to Equation~\ref{eq:main-inequality-equation}) are valid on an interval \( [t, t + \delta''] \) for some \( \delta'' > 0 \). Hence,it follows that the inequalities in Lemma~\ref{lem:boundedODE} hold for all \( t \geq 0 \), thus completing the proof.

\end{proof}We now demonstrate that the two lemmas above lead to Theorem \ref{thm:SBM}.
Suppose first that the two-agent system \( (a\beta, b, (x_{L}, x_{R})) \) has a solution \( (x_{1}(t), x_{2}(t)) \), and an equilibrium \( e = (e_{L}, e_{R}) \) such that \( X_{L} \times X_{R} \subseteq X_{eq}(a\beta, b) \). The equilibrium is given by
\[
e = (e_{L}, e_{R}) = \left(\frac{a\beta(c + d) + b d}{b + 2a\beta}, \frac{a\beta(c + d) + b c}{b + 2a\beta}\right),
\]
where \( c = f(x_R) \) and \( d = f(x_L) \).
According to the theorem's assumptions, there exists a small \( \delta > 0 \) and a sufficiently large \( N \) such that for all \( n \geq N \) and any \( \epsilon' > 0 \), we have 
\[
a_1 := \frac{a(1-\delta)q(n)}{(1+\delta)(p(n) + q(n))} \geq a \beta - \epsilon', \quad a_2 := \frac{a(1+\delta)q(n)}{(1-\delta)(p(n) + q(n))} \leq a \beta + \epsilon'.
\]
These constants \( a_1 \) and \( a_2 \) represent slight perturbations around \( a\beta \), ensuring that the dynamics of the two blocks are close to those of the two-agent system as \( n \) increases. This implies that for sufficiently small \( \delta \) and large \( n \), we have \( X_{L} \times X_{R} \subseteq X_{eq}(a_{i}, b) \) for \( i = 1, 2 \).

The opinion dynamics for the larger system can be described by the following differential equations, analogous to those used in the two-agent model:
\begin{equation}
\label{eq:opinion_dynamics_SBM}
\begin{aligned}
\dot{\bar{x}}_{L}'(t) &= a_{1} (\bar{x}_{R}'(t) - \bar{x}_{L}'(t)) + b(f(\bar{x}_{L}'(t)) - \bar{x}_{L}'(t)), \\
\dot{\bar{x}}_{R}'(t) &= a_{2} (\bar{x}_{L}'(t) - \bar{x}_{R}'(t)) + b(f(\bar{x}_{R}'(t)) - \bar{x}_{R}'(t)).
\end{aligned}
\end{equation}
These differential equations approximate the evolution of opinions in the two-block stochastic block model, where the influence of neighbors is averaged over the large number of agents. The initial conditions are \( \bar{x}_{L}'(0) = \bar{x}_{L}(0) \) and \( \bar{x}_{R}'(0) = \bar{x}_{R}(0) \), ensuring \( (\bar{x}_{L}(0), \bar{x}_{R}(0)) \in X_{eq}(a_i, b) \) for \( i = 1, 2 \).

According to standard differential equation theory, the equilibrium of the system described by these equations is 
\[
(e_{L}', e_{R}') = \left( \frac{a_1 c' + b d' + a_2 d'}{a_1 + b + a_2}, \frac{a_1 c' + b c' + a_2 d'}{a_1 + b + a_2} \right),
\]
where \( f(\bar{x}_{R}'(t)) = c' \) and \( f(\bar{x}_{L}'(t)) = d' \). These equilibrium points \( (e_{L}', e_{R}') \) are perturbed versions of the two-agent system equilibrium, but their structure remains close to the original due to the small changes in \( a_1 \) and \( a_2 \).

Moreover, the solution to the ODEs converges to the equilibrium. Thus, for all \( \delta' > 0 \), there exists a \( T' > 0 \) such that
\[
|\bar{x}_{L}'(t) - e_{L}'| \leq \delta'/2, \quad |\bar{x}_{R}'(t) - e_{R}'| \leq \delta'/2 \quad \text{for all} \quad t \geq T',
\]
and for small \( \epsilon \), \( |e_{L} - e_{L}'| \leq \delta'/2 \) and \( |e_{R} - e_{R}'| \leq \delta'/2 \), implying that
\[
|\bar{x}_{L}'(t) - e_{L}| \leq \delta' \quad \text{and} \quad |\bar{x}_{R}'(t) - e_{R}| \leq \delta' \quad \text{for all} \quad t \geq T' \quad \text{and small enough} \quad \epsilon.
\]

Since we assume that all initial opinions converge to the equilibrium in the two-agent system, this implies that the paths \( \bar{x}_{L}(t) \) and \( \bar{x}_{R}(t) \) will converge to this specific equilibrium for all \( t \). Therefore, \( \bar{x}_{L}(t) \) and \( \bar{x}_{R}(t) \) satisfy the differential equations \eqref{eq:opinion_dynamics_SBM} for sufficiently small \( \epsilon > 0 \). Consequently, 
\[
|\bar{x}_{L}(t) - e_{L}| \leq \delta' \quad \text{and} \quad |\bar{x}_{R}(t) - e_{R}| \leq \delta' \quad \text{for all} \quad t \geq T'.
\]
A similar reasoning applies for initial opinions \( \underline{x}_{L}(0) \) and \( \underline{x}_{R}(0) \), showing that 
\[
|\underline{x}_{L}(t) - e_{L}| \leq \delta' \quad \text{and} \quad |\underline{x}_{R}(t) - e_{R}| \leq \delta' \quad \text{for all} \quad t \geq T' \quad \text{and sufficiently small} \quad \epsilon.
\]

Given that \( f(x) \) is a Lipschitz function, the differences \( |c - c'| \) and \( |d - d'| \) are bounded by some constant \( L \) multiplied by the differences in \( e_R \) and \( e_R' \) (or their perturbed counterparts). Thus, small changes in \( c \) and \( d \) lead to proportionally small changes in \( c' \) and \( d' \), ensuring that the difference between \( e \) and \( e' \) remains small, on the order of \( \epsilon \) or \( \delta \). This indicates that as \( n \) increases, the opinions of the agents get arbitrarily close to the equilibrium opinions \( e_L \) and \( e_R \), and the error in the approximation decreases to zero. Therefore, \( L \times \epsilon \) goes to zero as well since all agents within each block converge to the opinion of a single representative agent by Lemma \ref{lem:boundedODE}, the perturbation \( L\epsilon \) is uniformly controlled across all agents, ensuring that \( L\epsilon \to 0 \) uniformly as \( t \to \infty \).

By applying Lemma~\ref{lem:boundedODE}, we conclude that for any \( \epsilon > 0 \), there exist constants \( N \), \( T' \), and \( \delta > 0 \) such that for all \( t \geq T' \), \( A_{n} \in C_{\delta, n} \), and \( \boldsymbol{x}_{0} \in (X_{L} \times X_{R})^{n} \) with \( n \geq N \), we have 
\[
|x_{i, n, A_{n}, \boldsymbol{x}_{0}}(t) - e_L| \leq \epsilon \quad \text{for every left agent} \quad i, 
\]
and
\[
|x_{j, n, A_{n}, \boldsymbol{x}_{0}}(t) - e_R| \leq \epsilon \quad \text{for every right agent} \quad j.
\]
This concludes the proof of the theorem.

\end{document}